\documentclass[11pt]{article}
\usepackage{amsmath,amssymb,amsfonts,amsthm,epsfig}
\usepackage[usenames,dvipsnames]{xcolor}
\usepackage{bm,xspace}
\usepackage{tcolorbox}
\usepackage{cancel}
\usepackage{fullpage}
\usepackage{liyang}
\usepackage{framed}
\usepackage{verbatim}
\usepackage{enumitem}
\usepackage{array}
\usepackage{multirow}
\usepackage{afterpage}
\usepackage{mathrsfs}
\usepackage{pifont} 
\usepackage{chngpage}
\usepackage[normalem]{ulem}
\usepackage{boxedminipage}
\usepackage{caption}

\newcommand{\gnote}[1]{\footnote{{\bf \color{blue}Guy}: {#1}}}

\def\colorful{1}

\ifnum\colorful=1

\newcommand{\red}[1]{{\color{red} {#1}}}

\newcommand{\gray}[1]{{\color{gray}{#1}}}

\fi
\ifnum\colorful=0

\newcommand{\red}[1]{{{#1}}}

\newcommand{\gray}[1]{{{#1}}}

\fi

\newcommand{\TopDownDTsize}{\textsc{TopDownDTSize}}
\newcommand{\BuildTopDownDT}{\textsc{BuildTopDownDT}}
\newcommand{\Tribes}{\textsc{Tribes}}
\newcommand{\Parity}{\textsc{Parity}}
\newcommand{\Threshold}{\textsc{Threshold}}

\newcommand{\Maj}{\textsc{Maj}}
\newcommand{\sens}{\mathrm{sens}}
\newcommand{\exact}{\mathrm{exact}}
\newcommand{\error}{\mathrm{error}}
\newcommand{\cost}{\mathrm{cost}}
\newcommand{\score}{\mathrm{score}}

\newcommand{\pparagraph}[1]{\bigskip \noindent {\bf {#1}}}

\begin{document}

\title{Top-down induction of decision trees: \\
rigorous guarantees and inherent limitations \vspace*{10pt}
}

\author{Guy Blanc \and \hspace{10pt} Jane Lange \vspace{8pt} \\
\hspace{6pt} {\small {\sl Stanford University}}
\and Li-Yang Tan}  
\date{\small{\today}}

\maketitle

\begin{abstract}
Consider the following heuristic for building a decision tree for a function $f : \zo^n \to \{\pm 1\}$.  Place the {\sl most influential variable} $x_i$ of $f$ at the root, and recurse on the subfunctions $f_{x_i=0}$ and $f_{x_i=1}$ on the left and right subtrees respectively; terminate once the tree is an $\eps$-approximation of $f$.   We analyze the quality of this heuristic, obtaining near-matching upper and lower bounds:  

\begin{itemize}[leftmargin=0.5cm]
\item[$\circ$] {\sl Upper bound:} For every $f$ with decision tree size $s$ and every $\eps \in (0,\frac1{2})$, this heuristic builds a decision tree of size at most $s^{O(\log(s/\eps)\log(1/\eps))}$. 
\item[$\circ$] {\sl Lower bound:} For every $\eps \in (0,\frac1{2})$ and $s \le 2^{\tilde{O}(\sqrt{n})}$, there is an $f$ with decision tree size $s$ such that this heuristic builds a decision tree of size $s^{\tilde{\Omega}(\log s)}$. 
\end{itemize}
We also obtain upper and lower bounds for monotone functions: $s^{O(\sqrt{\log s}/\eps)}$ and $s^{\tilde{\Omega}(\sqrt[4]{\log s }
)}$ respectively.  The lower bound disproves conjectures of Fiat and Pechyony (2004) and Lee~(2009).\medskip 

Our upper bounds yield new algorithms for properly learning decision trees under the uniform distribution.  We show that these algorithms---which are motivated by widely employed and empirically successful top-down decision tree learning heuristics such as ID3, C4.5, and CART---achieve provable guarantees that compare favorably with those of the current fastest algorithm (Ehrenfeucht and Haussler,~1989), and even have certain qualitative advantages. Our lower bounds shed new light on the limitations of these heuristics. 

Finally, we revisit the classic work of Ehrenfeucht and Haussler.  We extend it to give the first uniform-distribution proper learning algorithm that achieves polynomial sample and memory complexity, while matching its state-of-the-art quasipolynomial runtime. 

\end{abstract}

\thispagestyle{empty}

\newpage
\setcounter{page}{1}

\section{Introduction}

Consider the problem of constructing a {\sl decision tree} representation of a function $f : \zo^n \to \{\pm 1\}$, where the goal is to build a decision tree for $f$ that is as small as possible, ideally of size close to the optimal decision tree size of $f$.  Perhaps the simplest and most natural approach is to proceed in a {\sl top-down}, greedy fashion: 
\begin{enumerate}
\item Choose a ``good" variable $x_i$ to query as the root of the decision tree;
\item Build the left and right subtrees by recursing on the subfunctions $f_{x_i=0}$ and $f_{x_i=1}$ respectively. 
\end{enumerate}
 
This reduces the task of building a decision tree to that of choosing the root variable---i.e.~determining the {\sl splitting criterion} of this top-down heuristic. 
Intuitively, a good root variable should be one that is very ``relevant" and ``important" in terms of determining the value of $f$; it is reasonable to expect that querying such a variable first would reduce the number of subsequent queries necessary.   Our focus in this paper will be on a specific splitting criterion: {\sl influence}.

\begin{definition}[Influence]
\label{def:influence} 
The {\sl influence of the variable $x_i$ on a function $f : \zo^n \to \{\pm 1\}$} is defined to be \[ \Inf_i(f) \coloneqq \Prx_{\bx\sim \zo^n}[f(\bx) \ne f(\bx^{\oplus i})],\] 
where $\bx$ is drawn uniformly at random, and $\bx^{\oplus i}$ denotes $\bx$ with its $i$-th coordinate flipped.
\end{definition}

Influence is a fundamental and well-studied notion in the analysis of boolean functions~\cite{ODBook}.  It is the key quantity of interest in many landmark results (e.g.~the KKL inequality~\cite{KKL88}, Friedgut's junta theorem~\cite{Fri98}, the Invariance Principle~\cite{MOO10})~and open problems (e.g.~the Gotsman--Linial conjecture~\cite{GL89}, the Aaronson--Ambainis conjecture~\cite{AA14}, the Fourier Entropy-Influence conjecture~\cite{FK96})~of the field.  Beyond the analysis of boolean functions, this notion has been widely employed across both algorithms and complexity theory, where it has indeed proven to be a useful quantitative measure of the relevance and importance of a variable. Most relevant to the algorithmic applications in this paper, influence has been a key enabling ingredient in a large number of results in learning theory~\cite{BT96,BBL98,LMN93,Ser04,OS07,OW13,GS10,DRST10,Kane14-intersections,Kane14-GL,BCOST15}.

\subsection{Influence as a splitting criterion}  We now give a formal description of the heuristic for constructing decision trees that we study.  We define a {\sl bare tree} to be a decision tree with unlabeled leaves, and write $T^{\circ}$ to denote such trees.  We refer to any decision tree $T$ obtained from $T^{\circ}$ by a labelling of its leaves as a {\sl completion} of $T^{\circ}$.  Given a bare tree $T^\circ$ and a function $f$, there is a canonical completion of $T^\circ$ that minimizes the approximation error with respect to $f$: 

\begin{definition}[$f$-completion of a bare tree] 
Let $T^\circ$ be a bare tree and $f : \zo^n \to \{\pm 1 \}$.  Consider the following completion of $T^\circ$:  for every leaf $\ell$ in $T^\circ$, label it $\sign(\E[f_\ell(\bx)])$, where $f_\ell$ is the restriction of $f$ by the path leading to $\ell$ and $\bx\sim \zo^n$ is uniform random.  This completion minimizes the approximation error $\Pr[T(\bx) \ne f(\bx)]$, and we refer to it as the \emph{$f$-completion of $T^\circ$}. 
\end{definition} 

In addition to the function~$f$, our heuristic will also take in an error parameter $\eps$, allowing us to construct both {\sl exact} ($\eps = 0$) and {\sl approximate} ($\eps \in (0,\frac1{2})$) decision tree representations of $f$.  

\begin{figure}[H]
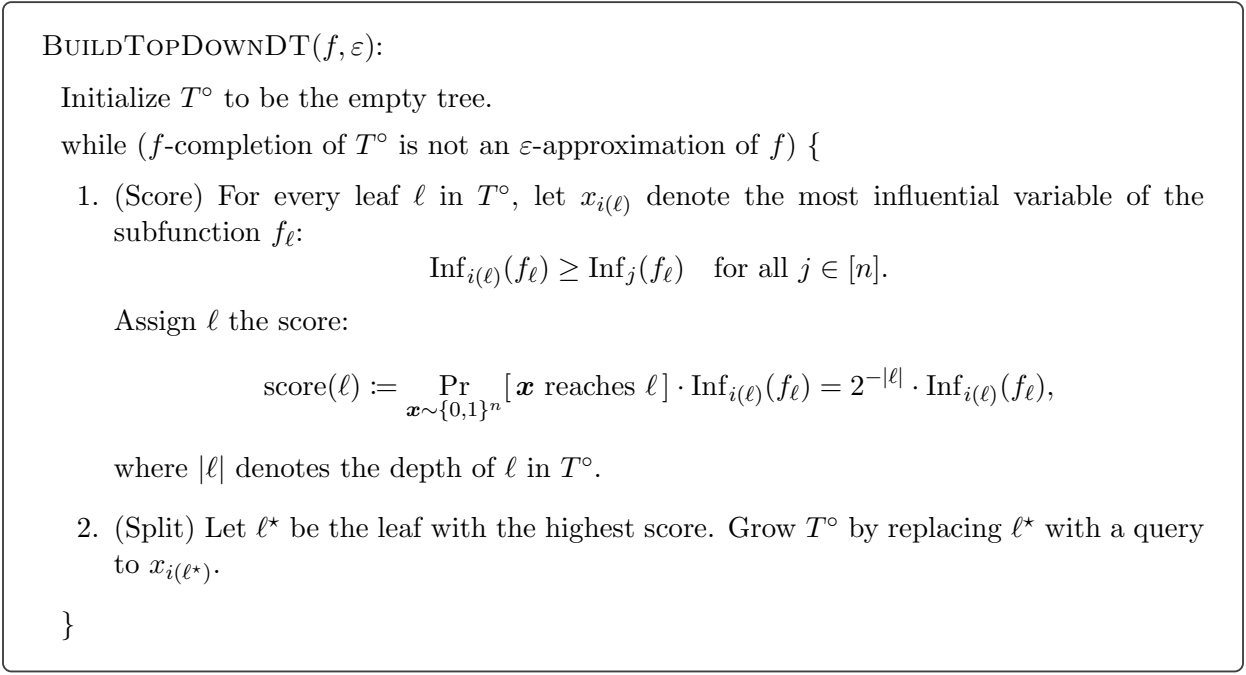

  \captionsetup{width=.9\linewidth}
\begin{tcolorbox}[colback = white,arc=1mm, boxrule=0.25mm]
\vspace{3pt} 

{\sc BuildTopDownDT}($f,\eps$):  \vspace{6pt} 

\ \ Initialize $T^\circ$ to be the empty tree. \vspace{4pt} 

\ \ while ($f$-completion of $T^\circ$ is not an $\eps$-approximation of $f$) \{
\vspace{-3pt} 
\begin{enumerate}
\item {(Score)} For every leaf $\ell$ in $T^\circ$, let $x_{i(\ell)}$ denote the most influential variable of the subfunction $f_\ell$: 
\[ \Inf_{i(\ell)}(f_\ell) \ge \Inf_j(f_\ell) \quad \text{for all $j\in [n]$.} \] 
Assign $\ell$ the score: 
\begin{align*}
 \mathrm{score}(\ell) &\coloneqq \Prx_{\bx \sim \zo^n}[\,\text{$\bx$ reaches $\ell$}\,] \cdot  \Inf_{i(\ell)}(f_\ell)  = 2^{-|\ell|} \cdot \Inf_{i(\ell)}(f_\ell),
 \end{align*} 
 where $|\ell|$ denotes the depth of $\ell$ in $T^\circ$. 
 \item {(Split)} Let $\ell^\star$ be the leaf with the highest score.  Grow $T^{\circ}$ by replacing $\ell^\star$ with a query to $x_{i(\ell^\star)}$. 
\end{enumerate}
\ \ \} \vspace{3pt}
\end{tcolorbox}
\caption{Top-down heuristic for building an $\eps$-approximate decision tree representation of~$f$, with influence as the splitting criterion.}
\label{fig:TopDown}
\end{figure}

In words, $\BuildTopDownDT$ builds a bare tree $T^\circ$ in a top-down fashion, starting from the empty tree.  In each iteration, we first check if the $f$-completion of $T^\circ$ is an $\eps$-approximation of $f$, and if so, we output the completion.  
Otherwise, we split the leaf $\ell^\star$ with the highest score by querying the most influential variable of $f_{\ell^\star}$, where the score of a leaf $\ell$ is the influence of the most influential variable of $f_\ell$ normalized by the depth of $\ell$ within $T^\circ$.\footnote{There are two possibilities for ties in $\BuildTopDownDT$: two variables may have the same influence within a subfunction $f_\ell$, and two leaves may have the same score.  Our upper bounds hold regardless of how ties are broken, and our lower bounds hold even if ties are broken in the most favorable way.} 

\subsection{This work}

By design, the decision tree returned by $\BuildTopDownDT(f,\eps)$ is an $\eps$-approximation of $f$.   We write $\TopDownDTsize(f,\eps)$ to denote the size of this tree, and when $\eps =0$, we simply write $\TopDownDTsize(f)$.  The question that motivates our work is: 
\smallskip

\begin{center}
{\sl What guarantees can we make on $\TopDownDTsize(f,\eps)$ as a function of} \vspace{2pt} \\
{\sl the optimal decision tree size of $f$ and $\eps$?}
\end{center}
\smallskip 

That is, we would like to understand the quality of $\BuildTopDownDT$ as a heuristic for constructing exact and approximate decision tree representations.  In addition to being a natural structural question concerning decision trees, this question also has implications in learning theory.  Indeed, $\BuildTopDownDT$ is motivated by top-down decision tree learning heuristics such as ID3, C4.5, and CART that are widely employed and empirically successful in machine learning practice.  We discuss the learning-theoretic context and applications of our structural results in~\Cref{sec:learning}, and the connection to practical machine learning heuristics in~\Cref{sec:ID3}.  

To our knowledge, the question above has not been studied in such generality.  The most directly relevant prior work is that of Fiat and Pechyony~\cite{FP04}, who considered the case when $f$ is either a linear threshold function or a read-once DNF formula, and the setting of exact representation ($\eps = 0$).  For such functions, they proved that the heuristic builds an exact decision tree representation of optimal size.   We give an overview of other related work in~\Cref{sec:related}.

\section{Our results}
\label{sec: our results} 

As our first contribution, we give near-matching upper and lower bounds that provide a fairly complete answer to the question above.    Our upper bound is as follows: 

\begin{restatable}[Upper bound for approximate representation]{theorem}{TDUpper}
\label{thm:TD-upper} 
For every $\eps \in (0,\frac1{2})$ and every size-$s$ decision tree $f$,  we have $\TopDownDTsize(f,\eps) \le s^{O(\log(s/\eps)\log(1/\eps))}$.
\end{restatable}

We complement~\Cref{thm:TD-upper} with lower bounds showing that (a) for exact representation ($\eps = 0$), no non-trivial upper bound can be obtained; and (b) for approximate representation ($\eps \in (0,\frac1{2})$), the dependence on $s$ in~\Cref{thm:TD-upper} is essentially optimal:

\begin{restatable}[Lower bounds for exact and approximate representations]{theorem}{TDlower}
\label{thm:TD-lower}\ 
\begin{enumerate}
\item[(a)] {\emph{Exact representation:}} There is an  $f : \zo^n \to \{ \pm 1\}$ with decision tree size $s = \Theta(n)$ such that $\TopDownDTsize(f) \ge 2^{\Omega(s)}$. 
\item[(b)] {\emph{Approximate representation:}} For every $\eps \in (0,\frac1{2})$ and function $s(n) \le 2^{\tilde{O}(\sqrt{n})}$, there is an $f: \{0,1\}^n \to \{\pm 1\}$ with decision tree size $s$ such that $\TopDownDTsize(f,\eps) \ge s^{\tilde{\Omega}(\log s)}$. 
\end{enumerate} 
\end{restatable}

Prior to our work, it was not known whether an upper bound of $\TopDownDTsize(f,\eps) \le \poly(s,1/\eps)$ holds for all size-$s$ decision trees $f$ and $\eps \in (0,\frac1{2})$;~\Cref{thm:TD-lower}(b) provides a strong negative answer.   Indeed, such an upper bound had been conjectured to hold for the class of {\sl monotone} functions~\cite{Lee09}.  We now discuss our results on monotone functions, which disprove this conjecture, along with a stronger variant of it for exact representation~\cite{FP04}. 

\pparagraph{Monotone functions.} 
A monotone boolean function $f : \zo^n \to \{\pm 1\}$ is one that satisfies $f(x) \le f(y)$ for all $x \preceq y$ (where $x \preceq y$ iff $x_i \le y_i$ for all $i\in [n]$).  An elementary and useful fact about monotone functions is that the influence of a variable on a monotone function $f$ is equivalent to its {\sl correlation} with $f$:

\begin{restatable}[Influence $\equiv$ correlation for monotone functions]{fact}{monoInfluence}
    \label{fact:mono-influence}
For all monotone functions $f : \zo^n \to \{ \pm 1\}$ and $i\in [n]$, we have $\Inf_i[f] =  2\E[f(\bx)\bx_i]-\E[f(\bx)]$.\footnote{The equivalence between influence and correlation for monotone functions is more transparent if one works with $\{\pm 1\}^n$ instead of $\zo^n$ as the domain: for monotone functions $f : \{\pm 1\}^n\to \{\pm 1\}$, we have $\Inf_i[f] = \E[f(\bx)\bx_i]$.}
\end{restatable}

Therefore, for monotone functions, splitting on the most influential variable of a subfunction is equivalent to splitting on the variable that has the {\sl highest correlation} with the subfunction.\footnote{\label{fn:parity}We observe that for general {\sl non-monotone} functions, correlation can in general be a very poor splitting criterion, in the sense of building a decision tree that is much larger than the optimal decision tree.  Consider $f : \zo^n \to \{ \pm 1 \}$ where $f(x) = x_j \oplus x_k$, the parity of two variables.  The optimal decision tree size of $f$ is $4$, but since $\E[f(\bx)\bx_i] = 0$ for all $i\in [n]$, the top-down heuristic using correlation as its splitting criterion may build a tree of size $\Omega(2^n)$ before achieving any non-trivial accuracy $\eps < \frac1{2}$.  (On the other hand, the top-down heuristic using influence as its splitting criterion would build the optimal tree of size $4$.)  We revisit this observation in~\Cref{sec:ID3}.}   

Our proof of~\Cref{thm:TD-upper} extends in a straightforward manner to give a different upper bound under the assumption of monotonicity, where the dependence on $s$ is significantly better.  We refer to a size-$s$ decision tree computing a monotone function as a size-$s$ monotone decision tree.   

\begin{restatable}[Upper bound for approximate representation of monotone functions.]{theorem}{TDUpperMonotone}
    \label{thm:TD-upper-monotone} 
    For every $\eps \in (0,\frac1{2})$ and every size-$s$ monotone decision tree $f$,  we have $\TopDownDTsize(f,\eps) \le s^{O(\sqrt{\log s}/\eps)}$.
\end{restatable}

In analogy with~\Cref{thm:TD-lower}, we also obtain lower bounds for exact and approximate representations of monotone functions: 

\begin{theorem}[Lower bounds for exact and approximate representations of monotone functions] \ 
\label{thm:TD-lower-monotone} 
\begin{enumerate}
\item[(a)] {\emph{Exact representation:}} There is a monotone  $f : \zo^n \to \{ \pm 1\}$ with decision tree size $s = \Theta(n)$ such that $\TopDownDTsize(f) \ge 2^{\Omega(s)}$. 
\item[(b)] {\emph{Approximate representation:}} For every $\eps \in (0,\frac1{2})$ and function $s(n) \le 2^{\tilde{O}(n^{4/5})}$, there is an $f: \{0,1\}^n \to \{\pm 1\}$ with decision tree size $s$ such that $\TopDownDTsize(f,\eps) \ge s^{\tilde{\Omega}(\sqrt[4]{\log s})}$. 
\end{enumerate} 
\end{theorem}

Although we have stated~\Cref{thm:TD-lower-monotone} in terms of the specific heuristic $\BuildTopDownDT$ that we study, the actual lower bounds that we establish are significantly stronger: they apply to all ``impurity-based top-down heuristics".  This is a broad class that captures a wide variety of decision tree learning heuristics used in machine learning practice, including ID3, C4.5, and CART; see~\Cref{sec:ID3} for details.

\begin{theorem}[Stengthening of~\Cref{thm:TD-lower-monotone}(b)] 
\label{thm:TD-lower-monotone-any}
For every $\eps \in (0,\frac1{2})$ and function $s(n) \le 2^{\tilde{O}(n^{4/5})}$, there is a size-$s$ monotone decision tree $f$ such that the $\eps$-approximator built by any impurity-based top-down heuristic must have size $s^{\tilde{\Omega}(\sqrt[4]{\log s})}$.
\end{theorem}

\smallskip

\noindent {\bf Disproving conjectures of Fiat--Pechyony and Lee.}   Motivated by applications in learning theory (discussed next in~\Cref{sec:learning}), Fiat and Pechyony~\cite{FP04} and Lee~\cite{Lee09} also considered the quality of $\BuildTopDownDT$ as a heuristic for building decision trees for monotone functions.

\cite{FP04} conjectured that for all monotone functions $f$, even in the case of exact representation ($\eps = 0$), $\BuildTopDownDT$ returns a tree of minimal depth and size ``not far from minimal." \Cref{thm:TD-lower-monotone}(a) provides a counterexample to the conjectured bound on size, and the function in~\Cref{thm:TD-lower-monotone}(b) disproves the conjecture about depth; see~\Cref{remark: depth separation}.\footnote{For clarity of exposition, throughout this overview we discuss our results with decision tree {\sl size} as the complexity measure.  There are analogues of all of our results, both upper and lower bounds, for decision tree {\sl depth} as the complexity measure. }  

Stated in the notation of our paper,~\cite{Lee09} raised the possibility that $\TopDownDTsize(f,\eps) \le \poly(s,1/\eps)$ for all size-$s$ monotone decision trees $f$ and $\eps \in (0,\frac1{2})$. The author further remarked that ``showing $\TopDownDTsize(f,\eps) \le \poly(s)$, even only for constant accuracy $\eps$,\footnote{That is, a bound of the form $\TopDownDTsize(f,\eps) \le s^{O_\eps(1)}$.} would be a huge advance".  \Cref{thm:TD-lower-monotone}(b) rules this out.  

\subsection{Algorithmic applications: Properly learning decision trees} 
\label{sec:learning} 

Learning decision trees has been a touchstone problem in uniform-distribution PAC learning for more than thirty years.  It sits right at the boundary of our understanding of efficient learnability, and continues to be the subject of intensive research.   The seminal work of Ehrenfeucht and Haussler~\cite{EH89} gave a $\poly(n^{\log s},1/\eps)$-time algorithm for learning decision trees using random examples (see also~\cite{Blu92} for an alternative proof based on Rivest's algorithm for learning decision lists~\cite{Riv87});\footnote{In fact, the algorithm of~\cite{EH89}  learns decision trees in the more challenging setting of {\sl distribution-free} PAC learning.  All other results in this discussed in this section, including ours, are specific to uniform-distribution learning, and we focus our exposition on this setting.} subsequently, Linial, Mansour, and Nisan~\cite{LMN93} gave an algorithm that also runs in quasipolynomial time, but achieves polynomial sample complexity; Kusilevitz and Mansour~\cite{KM93}, leveraging a novel connection to cryptography~\cite{GL89}, gave a polynomial-time algorithm using membership queries;  Gopalan, Kalai, and Klivans~\cite{GKK08} obtained an {\sl agnostic} analogue of~\cite{KM93}'s algorithm, extending it to tolerate adversarial noise; O'Donnell and Servedio~\cite{OS07} gave a polynomial-time algorithm for learning {\sl monotone} decision trees from  random examples; recent work of Hazan, Klivans, and Yuan~\cite{HKY18} gives an algorithm agnostically learning decision trees with {\sl polynomial sample complexity}; even more recent work of Chen and Moitra~\cite{CM19} gives an algorithm for learning {\sl stochastic} decision trees. 

\bigskip

\noindent {\bf Properly learning decision trees.}  When learning decision trees, it is natural to seek a  hypothesis that is itself a decision tree. Indeed, it may be natural to seek a decision tree hypothesis even when learning other concept classes.  The simple structure of decision trees makes them desirable both in terms of interpretability and explanatory power, which is why they are ubiquitous in empirical machine learning. A further advantage of decision tree hypotheses is that they are very fast to evaluate: evaluating a depth-$d$ decision tree on a given input takes time $O(d)$,\footnote{Every size-$s$ decision tree is well-approximated by a decision tree of depth $O(\log s)$.} whereas evaluating say a degree-$d$ polynomial---another canonical and ubiquitous representation class in learning theory---can take time $\Theta(n^d)$,  the number of monomials in the polynomial.

In learning theory, algorithms that return a hypothesis belonging to the concept class are known as {\sl proper}.   Understanding the complexity of proper learning (vis-\`a-vis improper learning) is an important research direction in learning theory~\cite{Fel16}; proper learning also has deep connections to proof complexity~\cite{ABFKP08} and property testing~\cite{GGR98}.  

\subsubsection{New proper learning algorithms}  

Among the decision tree learning algorithms discussed at the beginning of this subsection, the only one that is proper is the one of Ehrenfeucht and Haussler~\cite{EH89}. Our upper bounds on $\TopDownDTsize$ yield new algorithms for properly learning decision trees under the uniform distribution: 

\begin{theorem}[Algorithmic consequence of~\Cref{thm:TD-upper}]
\label{thm:learn-general} 
Size-$s$ decision trees can be properly learned under the uniform distribution in time $\poly(n,s^{\log(s/\eps)\log(1/\eps)})$ using membership queries.\footnote{We remark that our algorithm only requires fairly ``mild" use of membership queries.  Our algorithm only requires \emph{random edge samples} (\Cref{def:random edges}), and hence falls within both the random walk model of Bshouty et al.~\cite{BMOS05} and the local membership queries model of Awasthi et al.~\cite{AFK13}.  These (incomparable) models are natural relaxations of the standard model of learning from random examples, and do not allow the learning algorithm unrestricted membership query access to the target function.} 
\end{theorem} 

Analogously,~\Cref{thm:TD-upper-monotone} yields a new algorithm for learning  monotone decision trees using only random examples.  The learnability of monotone functions with respect to various complexity measures has been the subject of intensive study in uniform-distribution  learning~\cite{HM91,KV94,KLV94,Bsh95,BT96,BBL98,Ver98,SM00,Ser04,OS07,Sel08,DLMSWW09,Lee09,JLSW11,OW13}.

\begin{restatable}[Algorithmic consequence of~\Cref{thm:TD-upper,thm:TD-upper-monotone}]{theorem}{monotoneLearning}
\label{thm:learn-monotone} 
Size-$s$ monotone decision trees can be properly learned under the uniform distribution in time $\poly(n,\min(s^{O(\log(s/\eps)\log(1/\eps))}, s^{O(\sqrt{\log s}/\eps)}))$ using only random examples.
\end{restatable}

\noindent We now compare our results with the prior state of the art for properly learning decision trees.

\begin{itemize}[leftmargin=0.5cm]
\item[$\circ$]  {\bf Polynomial-time algorithms for superlogarithmic size.}  \Cref{thm:learn-general,thm:learn-monotone} give the first polynomial-time algorithms for properly learning decision trees of size $\omega(\log n)$ to constant accuracy.  To see this, we first note that~\cite{EH89}'s runtime of $\poly(n^{\log s},1/\eps)$ is superpolynomial time for any $s = \omega(1)$.   Alternatively, functions depending on $k \ll n$ variables (``$k$-juntas") can be properly learned in time $\poly(n,2^k)$, using random examples for monotone juntas, and membership queries otherwise~\cite{BL97,MOS04}. Since every size-$s$ decision tree certainly depends on at most $k \le s$ variables, this runtime is polynomial for decision trees of size $s = O(\log n)$, but becomes superpolynomial once $s = \omega(\log n)$.  In contrast, the runtimes of our algorithms in~\Cref{thm:learn-general,thm:learn-monotone} remain polynomial for $s = 2^{\Omega(\sqrt{\log n})}$ and $s = 2^{\Omega((\log n)^{2/3})}$ respectively. 

\item[$\circ$] {\bf Dimension-independent hypothesis size.}  Related to the above, the sizes of the hypotheses returned by the algorithms of~\Cref{thm:learn-general,thm:learn-monotone} are $s^{O(\log(s/\eps)\log(1/\eps))}$ and $s^{O(\sqrt{\log s}/\eps)}$ respectively, independent of $n$, whereas the size of the hypotheses returned by~\cite{EH89}'s algorithm can be as large as $n^{\Omega(\log s)}$.  This is gap can be exponential or even larger for small values of $s$.  

\item[$\circ$] {\bf Average depth as the complexity measure.} Our algorithms and analyses extend easily to accommodate {\sl average depth} as the complexity measure.  The average depth of a decision tree, $\triangle(T)$, is the number of queries $T$ makes on a uniform random input.  Average depth is a stronger complexity measure than size since $\triangle(T) \le \log(\mathrm{size}(T))$.\footnote{Furthermore, it is easy to construct examples of decision trees $T$ with the largest possible gap between these measures: $\triangle(T) = O(1)$ and $\log(\mathrm{size}(T)) = \Omega(n)$.}  

\begin{theorem}[Learning trees with small average depth] 
\label{thm: learning average depth trees}
Decision trees of average depth $\triangle$ can be properly learned under the uniform distribution in time $\poly(n,2^{\triangle^2/\eps})$ using membership queries, and monotone decision trees of average depth $\triangle$ can be properly learned in time $\poly(n,2^{\triangle^{3/2}/\eps})$ using random examples.
\end{theorem}

 To our knowledge, these represent the first polynomial-time algorithms for properly learning decision trees of superconstant average depth, $\triangle = \omega(1)$.  Prior to our work, the fastest algorithm ran in time $\poly(n^{\triangle/\eps})$; this algorithm, which uses random examples, follows implicitly from the results of Mehta and Raghavan~\cite{MR02}. 

\end{itemize}

\subsection{Proper learning with polynomial sample and memory complexity} 
\label{sec:memory}

For our final contribution, we revisit the classic algorithm of Ehrenfeucht and Haussler~\cite{EH89}.  As discussed above, this remains the fastest algorithm for properly learning decision trees.  We extend it to give the first uniform-distribution proper algorithm that achieves polynomial sample and memory complexity, while matching its state-of-the-art quasipolynomial runtime (\Cref{thm:our-proper}).

\vspace{5pt} 

\begin{table}[H]
  \captionsetup{width=.9\linewidth}
\begin{adjustwidth}{-2em}{}
\renewcommand{\arraystretch}{1.7}
\centering
\begin{tabular}{|c|c|c|c|c|}
\hline
  Reference  & Running time  & Sample complexity   & Memory complexity &  Proper? \\ \hline
\cite{EH89} & $\poly(n^{\log s},1/\eps)$ & $\poly(n^{\log s},1/\eps)$ & $\poly(n^{\log s},1/\eps)$ & $\checkmark$ \\ [.2em] \hline
 \cite{LMN93} & $ \poly(n^{\log(s/\eps)})$  & $\poly(s,1/\eps)\cdot \log n$ & $\poly(n,s,1/\eps)$ & $\times$ \\ [.2em]  \hline 
 \cite{MR02} & $\poly(n^{\log(s/\eps)})$ & $\poly(s,1/\eps)\cdot \log n$ & $\poly(n^{\log(s/\eps)})$ & $\checkmark$ \\ [.2em] 
 \hline \hline
  {\bf This work} & $\poly(n^{\log s},1/\eps)$  & $\poly(s,1/\eps)\cdot \log n$ & $\poly(n,s,1/\eps)$ & $\checkmark$ \\ [.2em] \hline
\end{tabular}
\caption{Algorithms for learning size-$s$ decision trees from random examples under the uniform distribution}
\label{table:EH-approx}
\end{adjustwidth}
\end{table}
\vspace{-7pt} 

Ehrenfeucht and Haussler had posed (as the first open problem of their paper) the question of achieving polynomial sample complexity.  Such algorithms were subsequently obtained by Linial, Mansour, and Nisan~\cite{LMN93} and Mehta and Raghavan~\cite{MR02}.  Interestingly, these two algorithms are very different from each other and from~\cite{EH89}: the algorithm of~\cite{LMN93}, being Fourier-based, is non-proper, whereas the algorithm of~\cite{MR02}, which uses dynamic programming, has a large memory footprint.  Furthermore, both algorithms have a quasipolynomial dependence on $1/\eps$ in their runtimes, rather than~\cite{EH89}'s polynomial dependence.

This state of affairs raises the question of whether there is a {\sl single} algorithm that achieves ``the best of \cite{EH89}, \cite{LMN93}, and~\cite{MR02}" in each of the four metrics discussed above; see~\Cref{table:EH-approx}.  We give such an algorithm in this work (\Cref{thm:our-proper}).  Our algorithm is a surprisingly simple modification of~\cite{EH89}'s algorithm, but our analysis is more involved.  At a high level, the idea is to terminate~\cite{EH89}'s algorithm early to achieve our improved sample and memory complexity.  However, incorporating this plan with the inherently bottom-up nature of~\cite{EH89}'s algorithm necessitates a delicate error analysis. (In particular,~\cite{EH89}'s algorithm is an Occam algorithm, whereas ours is not.)\footnote{Although our algorithm, like the others in~\Cref{table:EH-approx}, only uses random examples, to our knowledge there are no known membership query algorithms that achieves our guarantees.}\footnote{We note that it is possible to combine the ideas in~\cite{EH89} and~\cite{MR02} to give an algorithm that runs in $\poly(n^{\log(s/\eps)})$ time and has sample and memory complexity $\poly(s,1/\eps)\cdot \log n$ and $\poly(n,s,1/\eps)$ respectively.  We do not provide the details in this paper since our main result (\Cref{thm:our-proper}) achieves strictly better guarantees.} 

We remark that there is an ongoing flurry of research activity on the memory complexity of learning basic concept classes under the uniform distribution, with a specific focus on tradeoffs between memory and sample complexity~\cite{Sha14,SVW16,Raz17,KRT17,MM17,Raz18,MM18,BOGY18,GRT18,GRT19}. 

\section{Discussion and related work}

\subsection{Relationship to practical machine learning heuristics}
\label{sec:ID3} 

Our work is motivated in part by the tremendous popularity and empirical success of top-down decision tree learning heuristics in machine learning practice, such as ID3~\cite{Qui86}, its successor C4.5~\cite{Qui93}, and CART~\cite{Bre17}.  The data mining textbook~\cite{WFHP16} describes C4.5 as ``a landmark decision tree program that is probably the machine learning workhorse most widely used in practice to date".  In a similar vein, quoting Kearns and Mansour~\cite{KM99}, ``In experimental and applied machine learning work, it is hard to exaggerate the influence of top-down heuristics for building a decision tree from labeled sample data  [...]~Dozens of papers describing experiments and applications involving top-down decision tree learning algorithms appear in the machine learning literature each year".

We give a high-level description of how these heuristics work, using the framework of uniform-distribution learning.  As we will soon see, they serve as motivation for the heuristic that we study, $\BuildTopDownDT$ (\Cref{fig:TopDown}).  These heuristics grow a bare tree $T^\circ$ for a function $f : \zo^n \to \zo$ as follows.  Consider the progress measure
\[ \mathcal{H}(T^\circ) \coloneqq \sum_{\ell \in \mathrm{leaves}(T^{\circ})} \Prx_{\bx \sim \zo^n}[ \,\text{$\bx$ reaches $\ell$}\,] \cdot \mathscr{G}(\E[f_{\ell}]), \] 
where $\mathscr{G} : [0,1] \to [0,1]$ is known as the {\sl impurity function}, and encapsulates the splitting criterion of the heuristic.  This carefully chosen function is restricted to be concave, symmetric around $\frac1{2}$, and to satisfy $\mathscr{G}(0) = \mathscr{G}(1) = 0$ and $\mathscr{G}(\frac1{2}) = 1$. For example, $\mathscr{G}$ is the binary entropy function in ID3 and C4.5; CART uses $\mathscr{G}(p) = 4p(1-p)$, known as the {\sl Gini criterion};~\cite{KM99} studies the variant $\mathscr{G}(p) = 2\sqrt{p(1-p)}$.\footnote{The work of Dietterich, Kearns, and Mansour~\cite{DKM96} gives a detailed experimental comparison of various impurity functions.}  Writing $T^{\circ}_{\ell,i}$ to denote $T^{\circ}$ with its leaf $\ell$ replaced with a query to the variable $x_i$, these heuristics, in a single iteration, grow $T^{\circ}$ to $T^{\circ}_{\ell^\star,i^\star}$, where 
\begin{equation} 
\text{$(\ell^\star,i^\star)$ is the leaf-variable pair that maximizes
$\mathcal{H}(T^\circ) - \mathcal{H}(T^{\circ}_{\ell^\star,i^\star})$.} \label{eq:progress}
\end{equation}
We refer to any such top-down heuristic as an {\sl impurity-based} heuristic, and the progress measure $\mathcal{H}(T^\circ) - \mathcal{H}(T^{\circ}_{\ell^\star,i^\star})$ as the {\sl purity gain}.

\pparagraph{Inherent limitations of impurity-based heuristics.} It is easy to see (and has been well known~\cite{Kea96}) that impurity-based heuristics can, in general, fare very badly, in the sense of building a decision tree that is much larger than the optimal decision tree.   For example, consider $f(x) = x_j \oplus x_k$ for $j,k\in [n]$, the parity of two variables.  For such a target function, {\sl regardless of the choice of the impurity function $\mathscr{G}$}, splitting on any of the $n$ variables results in zero purity gain.  This is because $\E[f] = \E[f_{x_i=b}]$ for all $i\in [n]$ and $b \in \zo$.   Therefore, {\sl any} impurity-based heuristic may build a tree of size $\Omega(2^n)$ before achieving any non-trivial error $\eps < \frac1{2}$, whereas the size of the optimal tree of $f$ is only~$4$.  

One could exclude such ``parity-like" examples by considering only {\sl monotone} functions.  Monotonicity is a ubiquitous condition in machine learning since many data sets are naturally monotone in their attributes.  In the case of monotone functions, it can be shown that for {\sl any impurity function $\mathscr{G}$}, the variable split that results in the most progress in the sense of (\ref{eq:progress}), i.e.~the variable $x_i$ that maximizes the purity gain
\[ \mathscr{G}(\E[f]) - \lfrac1{2}(\mathscr{G}(\E[f_{x_i=0}])+ \mathscr{G}(\E[f_{x_i=1}])),\]
is precisely the most influential variable of $f$ (we prove this in~\Cref{sec:monotone-gain}; see~\Cref{prop: impurity gain is influence}).   In other words, in the case of monotone functions, $\BuildTopDownDT$ closely models impurity-based heuristics.  The works of Fiat and Pechyony~\cite{FP04} and Lee~\cite{Lee09} (recall our discussion following~\Cref{thm:TD-lower-monotone-any}) were explicitly motivated by this observation, as are our results on monotone functions (\Cref{thm:TD-upper-monotone,thm:TD-lower-monotone,thm:TD-lower-monotone-any,thm:learn-monotone}). 

 As we will show, our monotone lower bounds for $\BuildTopDownDT$ actually apply to all impurity-based heuristics (\Cref{thm:TD-lower-monotone-any}), regardless of the choice of the impurity function $\mathscr{G}$ (hence including ID3, C4.5, and CART).\footnote{Different impurity functions $\mathscr{G}$ lead to different {\sl orderings} of leaves to split, and hence result in different trees.}  Since one could argue that real-world data sets are unlikely to be ``parity-like", we view our monotone lower bounds as providing more robust (albeit still only theoretical) evidence of the limitations and potential shortcomings of the impurity-based top-down heuristics used in practice. 
\bigskip

\noindent {\bf Top-down versus bottom-up: from practice to theory?}  We find it especially intriguing that the algorithm of Ehrenfeucht and Haussler~\cite{EH89}---which as discussed, remains the fastest algorithm for properly learning decision trees with provable runtime guarantees---builds its hypothesis tree {\sl bottom up}, in exactly the {\sl opposite} order from the top-down heuristics used in practice.   It is  natural to ask if top-down heuristics can serve as inspiration for the design and analyses of fundamentally different algorithms for properly learning decision trees.

Our algorithmic upper bounds for $\BuildTopDownDT$ (\Cref{thm:learn-general,thm:learn-monotone}) provide affirmative answers, and as discussed above, these new algorithms even have certain qualitative advantages over~\cite{EH89}.   Our lower bounds (\Cref{thm:TD-lower,thm:TD-lower-monotone}), on the other hand, establish their inherent limitations.  They imply that  $\BuildTopDownDT$ is provably {\sl not} a polynomial-time algorithm for properly learning decision trees using membership queries, or a polynomial-time algorithm for properly learning monotone decision trees using random examples.  Either of these results would constitute a major advance in learning theory, and $\BuildTopDownDT$---and other impurity-based variants of it---had been a natural candidate for obtaining them. Indeed, the results of~\cite{Lee09} were explicitly motivated by the goal of showing that $\BuildTopDownDT$ {\sl is} a polynomial-time algorithm for properly learning monotone decision trees.  This is now ruled out by~\Cref{thm:TD-lower-monotone,thm:TD-lower-monotone-any}.\footnote{Blum et al.~\cite{BFJKMR94} gave an information-theoretic lower bound showing that no ``statistical query" algorithm can learn decision trees in polynomial time.  However, this lower bound does not apply when membership queries are allowed or when the function is assumed to be monotone.}

\subsection{Related work} 
\label{sec:related} 

Fiat and Pechyony~\cite{FP04} considered linear threshold functions and read-once DNF formulas, and showed that $\BuildTopDownDT$, when run on such functions, returns a decision tree of optimal size computing them exactly. (Stated in the notation of~\Cref{thm:TD-upper}, $\TopDownDTsize(f) = s$ for such functions.)  

Kearns and Mansour~\cite{KM99} (see also~\cite{Kea96,DKM96}) showed that impurity-based heuristics are {\sl boosting algorithms}, where one views the functions labeling internal nodes of the tree (single variables in our case) as weak learners.  At a high level, the proofs of our upper bounds (\Cref{thm:TD-upper,thm:TD-upper-monotone}) are similar in spirit to their analysis, in the sense that they are all incremental in nature, showing that each split contributes to the accuracy of the decision tree hypothesis.  However, our results and analyses are incomparable---for example,~\cite{KM99} does not relate the size of the resulting hypothesis to the size of the optimal decision tree;~\cite{KM99}'s analysis assumes the existence of weak learners for all filtered-and-rebalanced versions of the target distribution, whereas we carry out the entirety of our analyses with respect to the uniform distribution.\footnote{Indeed,~\cite{KM99}'s results concern impurity-based heuristics, and as discussed above, statements like~\Cref{thm:TD-upper} that apply to all functions cannot hold for such heuristics because of parity-like functions.}  

Recent work of Brutzkus, Daniely, and Malach~\cite{BDM19} studies a variant of ID3 proposed by~\cite{KM99}, focusing on learning conjunctions and read-once DNF formulas under product distributions.  They provide theoretical and empirical evidence showing that for such functions, the size-$t$ tree grown by~\cite{KM99}'s variant of ID3 achieves optimal or near-optimal error among all trees of size~$t$.  Concurrent work by the same authors~\cite{BDM19-smooth} shows that ID3 efficiently learns $(\log n)$-juntas in the setting of smoothed analysis.

\section{Preliminaries} 

Throughout this paper, we use bold font (e.g.~$\bx$ and $\bS$) to denote random variables; all probabilities and expectations are with respect to the uniform distribution unless otherwise stated.  

For any decision tree $T$, we say the \textit{size} of $T$ is the number of leaves in $T$, and the \textit{depth} of $T$ is length of the longest path between the root and a leaf. If a tree has size $1$, then it contains a single leaf, computes either the constant $+1$ or constant $-1$ function, and has depth $0$. For a function $f:\{0,1\}^n \rightarrow \{\pm 1\}$, the \textit{optimal decision tree size} of $f$ is the smallest $s$ for which there exists a decision tree of size $s$ that exactly computes $f$, and we write $\size(f)$ to denote this quantity. If $T$ is a decision tree that computes $f$, then we will often use $T$ interchangeably with $f$.

Choose any $f,g:\{0,1\}^n \rightarrow \{\pm 1\}$. Then, the \textit{error} is defined as
\begin{align*}
    \text{error}(f,g) = \underset{\boldsymbol{x} \sim \{0,1\}^n}{\text{Pr}}[f(\boldsymbol{x}) \neq  g(\boldsymbol{x})].
\end{align*}
We say that $f$ is an $\varepsilon$-approximation of $g$ if $\text{error}(f,g) \leq \varepsilon$. If $T^\circ$ is a bare tree, then $\text{error}(T^\circ, f)$ is shorthand for $\text{error}(T, f)$ where $T$ is the $f$-completion of $T^\circ$. We also use the following shorthand.
\begin{align*}
    \text{error}(f, \pm 1) = \min(\text{error}(f, - 1),\text{error}(f,  1)).
\end{align*}
The \textit{variance} of $f:\{0,1\}^n \rightarrow \{\pm 1\}$, denoted $\text{Var}(f)$, is
\begin{align*}
    \text{Var}(f) = 4 \cdot \Pr[f(\boldsymbol{x}) = -1] \cdot \Pr[f(\boldsymbol{x}) = 1].
\end{align*}
The \textit{total influence} of $f$, denoted $\text{Inf}(f)$, is
\begin{align*}
    \text{Inf}(f) = \sum_{i=1}^n \text{Inf}_i(f).
\end{align*}
It is easy to see that for any decision tree $T:\{0,1\}^n \rightarrow \{\pm 1\}$,
\begin{align*}
    \text{error}(T, \pm 1) \leq \text{Inf}(T)
\end{align*}
and
\begin{equation*}
    \frac{\text{Var}(T)}{2} \leq \text{error}(T, \pm 1) \leq \text{Var}(T)
\end{equation*}
always hold.

\newcommand{\topdown}{\BuildTopDownDT}

\section{Upper bounds on $\TopDownDTsize$: Proofs of \Cref{thm:TD-upper,thm:TD-upper-monotone}} 

Recall that \topdown{}$(f, \varepsilon)$ continually grows a bare tree, $T^\circ$, until the $f$-completion of $T^\circ$ is an $\varepsilon$-approximation of $f$. At a high level, the proofs of our upper bounds on \TopDownDTsize{} proceed as follows.
\begin{enumerate}[align=left]
    \item[\Cref{subsection: cost}] We define a progress metric, the ``cost" of $T^\circ$, which upper bounds the error of the $f$-completion of $T^\circ$ with respect to $f$. Hence, when the ``cost" drops below $\varepsilon$, \topdown{} can terminate. We show that whenever $\topdown{}$ grows $T^\circ$, the ``cost" of $T^\circ$ decreases by exactly the score of the leaf selected.
    \item[\Cref{subsection: lower bound on score}] We lower bound the score of the leaf that $\topdown{}$ selects.
    \item[\Cref{subsection: upper bound proofs}] We put the above together to prove upper bounds on \TopDownDTsize{}. At each step, the ``cost" of $T^\circ$ must decrease by at least the lower bounds in \Cref{subsection: lower bound on score}, which allows us to upper bound the number of steps until the ``cost" falls below $\varepsilon$. This is sufficient since the size of the tree that $\topdown{}$ produces is exactly one more than the number of steps it takes.
\end{enumerate}

\subsection{Definition and properties of ``Cost"}
\label{subsection: cost}

\begin{definition}[Cost of a bare tree]
    Let $f:\{0,1\}^n \rightarrow \{\pm 1\}$ be a function and $T^\circ$ be a bare tree. Then the \emph{cost of $T^\circ$ relative to $f$} is defined as
    \begin{align*}
        \cost_f(T^\circ) = \sum_{\text{leaf }\ell \in T^\circ} 2^{-|\ell|}\cdot \Inf(f_\ell).
    \end{align*}
\end{definition}

This cost function is useful to track because it naturally decreases during \textsc{BuildTopDownDT} and upper bounds the error of the completion.
\begin{lemma}[Properties of cost of a bare tree]
\label{lemma: cost properties}
    For any $f:\{0,1\}^n \rightarrow \{\pm 1\}$ and bare tree $T^\circ$, the following hold: 
    \begin{enumerate}
        \item \label{cost prop 1} $\error(T^\circ, f) \leq \cost_f(T^\circ)$.
        \item \label{cost prop 2} Choose any leaf $\ell$ of $T^\circ$ and variable $x_i$. Let $(T^\circ)'$ be the bare tree that results from replacing $\ell$ in $T^\circ$ with a query to $x_i$. Then,
        \begin{align*}
            \cost_f((T^\circ)') = \cost_f(T^\circ) - 2^{-|\ell|} \cdot \Inf_i(f_\ell).
        \end{align*}
    \end{enumerate}
\end{lemma}
At each step, $\textsc{BuildTopDownDT}$ splits the leaf with the largest score, resulting in the cost decreasing by exactly the score selected. Once the cost decreases to below $\varepsilon$, we know the completion of $T^\circ$ is an $\varepsilon$-approximation of $f$, meaning \textsc{BuildTopDownDT} can terminate.
\begin{proof}
    The proof of (\ref{cost prop 1}) is a simple application of the fact that $\error(g, \pm 1) \leq \Inf(g)$ for any boolean function $g$:
    \begin{align*}
         \text{error}(T^\circ, f)  &= \underset{\boldsymbol{x} \sim \{0,1\}^n}{\text{Pr}}[(\text{Completion of }T^\circ)(\boldsymbol{x}) \neq f(\boldsymbol{x})] \\
        &= \sum_{\text{leaf }\ell \in T^\circ} \underset{\boldsymbol{x} \sim \{0,1\}^n}{\text{Pr}}[\boldsymbol{x}\text{ reaches }\ell] \cdot \text{error}(f_\ell, , \pm 1) \\
        &\leq \sum_{\text{leaf }\ell \in T^\circ} 2^{-|\ell|} \cdot \text{Inf}_i(f_\ell) = \text{cost}_f(T^\circ).
    \end{align*}
   
    The proof of (\ref{cost prop 2}) follows from the fact that if $T$ is a tree with $x_i$ at the root, $T_0$ as its $0$-subtree, and $T_1$ as its $1$-subtree, then $\text{Inf}(T) - \text{Inf}_i(T) = \frac{1}{2}(\text{Inf}(T_0) + \text{Inf}(T_1))$. This fact is true because
     \begin{align*}
        \text{Inf}(T) - \text{Inf}_i(T) &= \sum_{j \neq i} \text{Inf}_j(T) \\
        &= \sum_{j \neq i} \frac{1}{2}\text{Inf}_j(T_0) + \frac{1}{2}\text{Inf}_j(T_1) \\
        &= \frac{1}{2} \left( \sum_{j=1}^n \text{Inf}_j(T_0) + \sum_{j =1}^n \Inf_j(T_1)\right) \\
        &=\frac{1}{2}(\text{Inf}(T_0) + \text{Inf}(T_1)). \qedhere
    \end{align*}
\end{proof}

\subsection{Lower bounds on the score of the leaf \topdown{} selects}
\label{subsection: lower bound on score}

We give two different lower bounds. These lower bounds are incomparable, so when proving \Cref{thm:TD-upper,thm:TD-upper-monotone}, we use whichever is better. Both of these lower bounds rely on a powerful inequality from the analysis of boolean functions due to O'Donnell, Saks, Schramm, and Servedio~\cite{OSSS05}, which we restate in the form most convenient for us.
\begin{theorem}
[Corollary of Theorem 1.1 from \cite{OSSS05}]
    \label{thm: osss}
    Let $f$ be a size-$s$ decision tree. Then,
    \begin{align*}
        \max_i \big(\Inf_i(f)\big) \geq \frac{\Var(f)}{\log s}.
    \end{align*}
\end{theorem}

We prove our first lower bound on the score of the leaf selected.
\begin{lemma}
    \label{lemma: minimum score additive}
    Let $f$ be a size $s$ decision tree. At step $j$, $\topdown{}(f, \varepsilon)$ selects a leaf, $\ell^*$ with score at least
    \begin{align*}
        \score(\ell^*) \geq \frac{\varepsilon}{(j+1) \log(s)}.
    \end{align*}
\end{lemma}
\begin{proof}
   If \textsc{BuildTopDownDT} has not terminated at step $j$, then, the completion of $T^\circ$ is not an $\varepsilon$-approximation of $f$. Equivalently,
    \begin{align*}
         \sum_{\text{leaf }\ell \in T^\circ} 2^{-|\ell|}\cdot \text{error}(f_\ell, \pm 1) > \varepsilon
    \end{align*}
    At step $j$, there are exactly $j+1$ leaves in $T^\circ$, so there must be at least one leaf, $\ell$, where
    \begin{align*}
        2^{-|\ell|}\cdot \error(f_\ell, \pm 1) > \frac{\varepsilon}{j+1}.
    \end{align*}
    Since $\Var(f_\ell) \geq \error(f_\ell, \pm 1) $, we also know
    \begin{align*}
        2^{-|\ell|}\cdot \Var(f_\ell) > \frac{\varepsilon}{j+1}.
    \end{align*}
    By~\Cref{thm: osss}, we know that there is some variable $x_i$ such that $\Inf_i(f_\ell) \geq \Var(f_\ell)/\log(\size(f_\ell))$. The optimal size of any restriction of $f$ is certainly at most the optimal size of $f$ itself, so
    \begin{align*}
        2^{-|\ell|}\cdot \Inf_i(f_\ell) > \frac{\varepsilon}{(j+1)\log(s)}.
    \end{align*}
    Since \textsc{BuildTopDownDT} picks a leaf with maximum score, and $\ell$ has a score at least $\eps/{(j+1)\log(s)}$, it must pick a leaf with at least that score.
\end{proof}

A standard fact from the analysis of boolean functions gives a $\log s$ upper bound on the total influence of a size-$s$ decision tree (see e.g.~\cite{OS07}). In order to prove a second lower bound on the score of the leaf that \topdown{} selects, we will need a refinement of this bound that takes into account the variance of the function.  The following lemma is a slight variant of a related (though incomparable) result in~\cite{BT15}, which upper bounds the total influence of an $s$-term DNF formula by $2\mu \log(s/\mu)$, where $\mu \coloneqq \Pr[f(\bx) = 1]$.  

\begin{lemma}[Total influence of size-$s$ DTs]
\label{lem:total-inf-DTs}
Let $f : \zo^n \to \{ \pm 1\}$ be computed by a size-$s$ decision tree $T$.  Then 
\[ \Inf(f) \le \Var(f) \log(4s/\Var(f)). \] 
\end{lemma} 

\begin{proof} 
We may assume without loss of generality that $\mu \coloneqq \Pr[f(\bx) = 1] \le \frac1{2}$, since $\Inf(f) = \Inf(\neg f)$ and if $f$ is a size-$s$ decision tree then so is its negation $\neg f$. Since 
\begin{align*} 
\Inf(f) &= \Ex_{\bx \sim\zo^n}[ \sens_f(\bx)]   \tag*{(where $\sens_f(x) \coloneqq |\{i \in [n] \colon f(x)\ne f(x^{\oplus i})\}|$)} \\
&= 2\cdot \Ex\big[\sens_f(\bx)\Ind[f(\bx)=1]\big]  \\
&\le 2 \sum_{\text{$1$-leaves $\ell \in T$}} 2^{-|\ell|} \cdot |\ell| \tag*{($\sens_f(x) \le |\ell|$ for every $x$ that reaches $\ell$)} \\
&\le 2\mu \log(s/\mu)  \tag*{(Concavity of $t \mapsto t\log(1/t)$, and $\mathrm{size}(T)\le s$)} \\
&\le \Var(f) \log(4s/\Var(f)) \tag*{($\Var(f) = 4\mu(1-\mu)$, and our assumption that $\mu \le \frac1{2}$)}, 
\end{align*}  
the lemma follows.
\end{proof} 

We now provide a second lower bound on the score of the leaf \topdown{} selects. The lower bound provided below in~\Cref{lemma: minimum score multiplicative} is better than the bound provided by~\Cref{lemma: minimum score additive} when $\text{cost}_f(T^\circ)$ is large.

\begin{lemma}
    \label{lemma: minimum score multiplicative}
    Let $f$ be a size $s$ decision tree. Suppose, at step $j$, that $\topdown{}(f, \varepsilon)$ has already constructed the bare tree $T^\circ$ and that $\text{cost}_f(T^\circ) \geq \varepsilon \log(4s/\eps)$ Then, the next leaf, $\ell^*$, that \topdown{} picks has score at least
    \begin{align*}
        \score(\ell^*) \geq \frac{\cost_f(T^\circ)}{(j+1)\log(4s/\eps)\log(s)}.
    \end{align*}
\end{lemma}

\begin{proof}
     We will show that when $\text{cost}_f(T^\circ)$ is large, there is some leaf with high total influence, which means it must have high variance, and finally a variable with high influence.

    We define:
    \begin{align*}
        h_s: [0,1] \rightarrow \R \text{ where } h_s(t) = t \log\left(\frac{4s}{t}\right) \text{ and } h_s(0) = 0.
    \end{align*}
    Then, for any tree $T$ of size at most $s$, we have that
    \begin{align*}
        \text{Inf}(T) \leq h_s(\text{Var}(T)).
    \end{align*}
    As long as $s \geq 1$, $h_s$ is an increasing concave function. This means it has a convex inverse, $h_s^{-1}$, and that for any tree $T$ of size at most $s$, the following lower bounds the variance.
    \begin{figure}[tb]
      \captionsetup{width=.9\linewidth}
        \centering
        \includegraphics[width=0.4\linewidth]{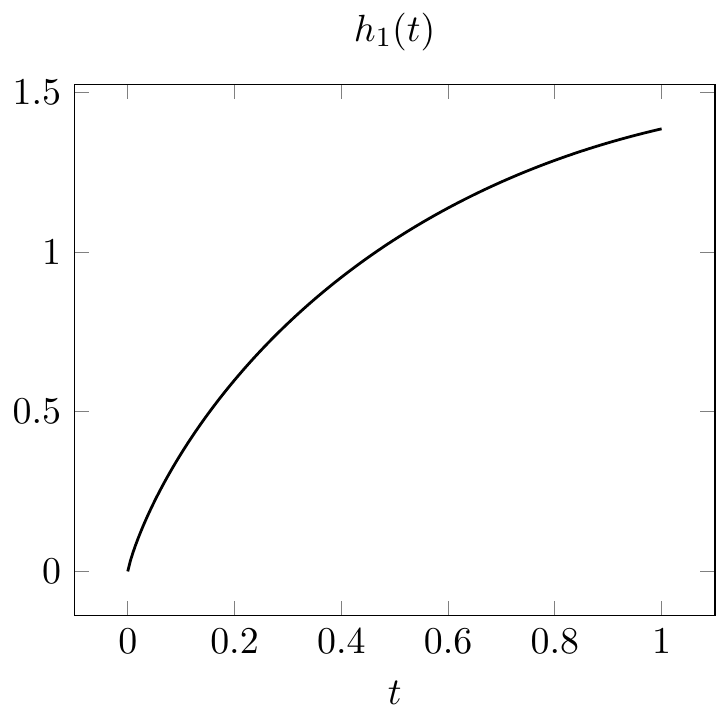}
         \includegraphics[width=0.4\linewidth]{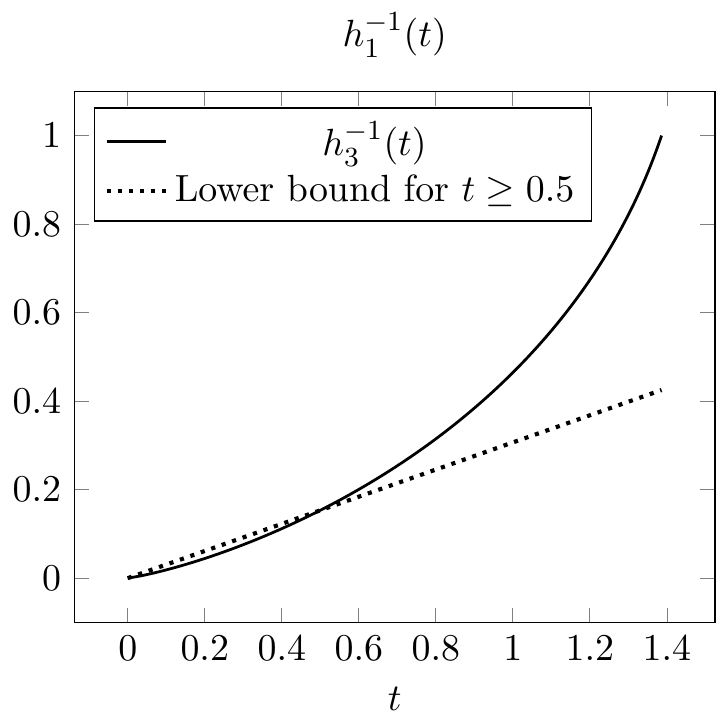}
        \caption{Graphs of the function $h_1(t) = t \cdot \log(\frac{4}{t})$ on the left, and of its inverse, $h_1^{-1}(t)$ on the right. Since the inverse is convex, we can use a linear lower bound as the dotted line in the right plot shows.}
        \label{fig:h and h inverse}
    \end{figure}
    
    \begin{align}
    \label{eq: bound variance by total influence}
        \text{Var}(T) \geq h_s^{-1}(\text{Inf}(T)).
    \end{align}
    Since $h_s^{-1}$ is convex and $h_s^{-1}(0) = 0$, we can lower bound it as follows. Choose arbitrary $a \in \mathbb{R}$. Then, for $t \geq a$ we have that $h_s^{-1}(t) \geq t \cdot \frac{h_s^{-1}(a)}{a}$. Choosing $a = \varepsilon\log(4s/\eps)$, we have that,
    \begin{align*}
        h_s^{-1}(t) \geq  \frac{t}{\log( 4s/\eps)} \text{ for all } t \geq \varepsilon\log\left( \frac{4s}{\varepsilon}\right).
    \end{align*}
    Consider the bare tree, $T^\circ$, at step $j$. By definition, it has cost
    \begin{align*}
        \sum_{\text{leaf }\ell \in T^\circ} 2^{-|\ell|}\cdot \text{Inf}(f_\ell) = \text{cost}_f(T^\circ).
    \end{align*}
    We next apply Jensen's inequality.
    \begin{align*}
        \sum_{\text{leaf }\ell \in T^\circ} 2^{-|\ell|}\cdot h_s^{-1}(\text{Inf}(f_\ell)) \geq h_s^{-1}(\text{cost}_f(T^\circ)).
    \end{align*}
    Since, at step $j$, there are $j+1$ leaves in $T^\circ$, for at least one of the leaves, $\ell$,
    \begin{align*}
    \label{eq: jensens applied to h}
        2^{-|\ell|}\cdot h_s^{-1}(\text{Inf}(f_{\ell})) \geq \frac{h_s^{-1}(\text{cost}_f(T^\circ))}{j+1} \geq \frac{\text{cost}_f(T^\circ)}{(j+1)\log(\frac{4s}{\varepsilon})}.
    \end{align*}
    By~\Cref{eq: bound variance by total influence}, we can lower bound the variance of $f_\ell$:
    \begin{align*}
        2^{-|\ell^*|} \cdot \text{Var}(f_{\ell}) \geq  2^{-|\ell|}\cdot h_s^{-1}(\text{Inf}(f_{\ell})) \geq \frac{\text{cost}_f(T^\circ)}{(j+1) \log(\frac{4s}{\varepsilon})}.
    \end{align*}
    Then, using~\Cref{thm: osss} and the fact that if $f$ is exactly computed by a size $s$ tree, then $f_{\ell}$ is exactly computed by a tree of size at most $s$.
    \begin{align*}
        2^{-|\ell|} \cdot \max_i\big(\Inf_i[f_\ell]\big) \geq \frac{\text{cost}_f(T^\circ)}{(j+1)\log(\frac{4s}{\varepsilon}) \log(s)}.
    \end{align*}
    Recall that \topdown{} picks the leaf with largest score, so it will pick a leaf with score at least $\frac{\text{cost}_f(T^\circ)}{(j+1)\log(4s/\eps) \log(s)}$.
\end{proof}

\subsection{Proofs of \Cref{thm:TD-upper,thm:TD-upper-monotone}}
\label{subsection: upper bound proofs}

Armed with the above Lemmas, we are now ready to prove our upper bounds on the size of the tree that \topdown{} produces.

\TDUpper*

\begin{proof}
    We use $C_j$ to refer to $\text{cost}_f(T^\circ)$ after $j$ steps of \topdown{}. The size of the tree returned is one more than the number of steps \topdown{} takes. Furthermore, if $C_j \leq \varepsilon$, then $T^\circ$ has error at most $\varepsilon$ at step $j$, so \topdown{} will return a tree of size at most $j + 1$. 
    
    Our analysis proceeds in two phases:
    \begin{itemize}[align=left]
        \item[\textbf{Phase 1:}] We will show that the larger $C_j$ is, the faster it must decrease at each step. This multiplicative reduction of $C_j$ will allow us to conclude that after at most $k =  s^{{\log(4s/\varepsilon) \log(1/\varepsilon)}}$ steps, that $C_k \leq \varepsilon\log(\frac{4s}{\varepsilon})$.
        \item[\textbf{Phase 2:}] We will argue that $C_j$ makes additive progress towards $0$ once it is less than $\varepsilon\log(\frac{4s}{\varepsilon})$, showing that after $m = s^{2\log(4s/\varepsilon) \log(1/\varepsilon)}$ steps, that $C_m \leq \varepsilon$.
    \end{itemize}
    Once $C_m \leq \varepsilon$, the algorithm must terminate.\\

    \noindent\textbf{Phase 1}: Based on~\Cref{lemma: minimum score multiplicative}, we know that during phase 1, \topdown{} will select a leaf with influence at least $\frac{\cost_f(T^\circ)}{(j+1)\log(4s/\eps) \log s}$ at each step $j$. From~\Cref{lemma: cost properties}, we know that:
    
    \begin{align*}
        C_{j} &\leq C_{j-1} - \frac{C_j}{j\log(4s/\eps) \log s} \\
        &= C_{j-1} \cdot \bigg(1 - \frac{1}{j\log(4s/\eps) \log s}\bigg).
    \end{align*}
    We can use this to bound $C_k$, the cost after some ($k$) number of steps, in terms of $C_0$.
    \begin{align*}
        C_k &\leq C_0 \prod_{j=1}^k \left(1 - \frac{1}{j\log(4s/\eps) \log s}\right)  \\
        &= C_0 \exp\Bigg(\sum_{j=1}^k \log\bigg(1 - \frac{1}{j\log(4s/\eps) \log s}\bigg)\Bigg).
    \end{align*}
    Using the fact that $\log(1 + t) < t$,
    \begin{align*}
        C_k &\leq C_0\exp\bigg(-\sum_{j=1}^k \frac{1}{j\log(4s/\eps) \log s}\bigg) \\
        &\leq C_0\exp\bigg(- \frac{\log k}{\log(4s/\eps) \log s}\bigg).
    \end{align*}
    We know that $C_0 \leq \log s$ because a size-$s$ decision tree has total influence at most $\log s$ (see e.g.~\cite{OS07}). Choosing
    \begin{align*}
        k = \exp\big(\log(4s/\eps) \log(s)\log(1/\eps)\big) = s^{\log(4s/\eps)\log(1/\eps)}
    \end{align*}
    it must be true that $C_k \leq \varepsilon \log(4s/\eps)$.\\

    \noindent\textbf{Phase 2}. This phase combines~\Cref{lemma: minimum score additive,lemma: cost properties}, which together imply that
    \begin{align*}
        C_{j+1} \leq C_j - \frac{\varepsilon}{(j+1) \log s}.
    \end{align*}
    This means that, for $m > k$,
    \begin{align*}
        C_k - C_m \geq \sum_{j=k+1}^{m} \frac{\varepsilon}{(j+1) \log s} \geq \frac{ \varepsilon}{\log s} (\log m - \log k).
    \end{align*}
    We are guaranteed to terminate at the first $j$ such that $C_j \leq \varepsilon$, or earlier.  Choosing
    \begin{align*}
        \log m = \frac{\log s}{\varepsilon} C_k + \log k
    \end{align*}
    ensures that $C_m \leq 0$, which means \topdown{} must terminate before step $m$. Plugging in $C_k \leq \varepsilon \log(4s/\eps)$ and $k = s^{\log(4s/\eps)\log(1/\eps)}$ gives that
    \begin{align*}
        m \leq s^{2 \log(4s/\eps) \log(1/\eps)}.
    \end{align*}
    Since \topdown{} terminates after at most $m$ steps, it returns a tree of size at most $m + 1$.
\end{proof}

The proof of \Cref{thm:TD-upper-monotone} is mostly the same as Phase 2 from the proof of \Cref{thm:TD-upper}, except we have a better guarantee on the starting cost.  We will use the following upper bound on the total influence of monotone decision trees, due to O'Donnell and Servedio~\cite{OS07}: 

\begin{theorem}[\cite{OS07}]
\label{thm:OS} 
Let $f$ be a size-$s$ monotone decision tree.  Then $\Inf(f) \le \sqrt{\log s}$. 
\end{theorem}

\TDUpperMonotone* 

\begin{proof}
    We use $C_j$ to refer to $\text{cost}_f(T^\circ)$ after $j$ steps of \topdown{}. By combining \Cref{lemma: minimum score additive} and \Cref{lemma: cost properties}, we know that
    \begin{align*}
        C_{j+1} \leq C_j - \frac{\varepsilon}{(j+1) \log s}.
    \end{align*}
    At any step $k$,
    \begin{align*}
        C_0 - C_k \geq \sum_{j=0}^{k-1} \frac{\varepsilon}{(j+1) \log s} \geq \frac{ \varepsilon \cdot \log k}{\log s}. 
    \end{align*}
    Since $f$ is a monotone decision tree of size $s$, it has total influence at most $\sqrt{\log s}$ (\Cref{thm:OS}). This means that $C_0 \leq \sqrt{\log s}$. We choose
    \begin{align*}
        k = \exp(\log(s)^{1.5}/ \varepsilon) = s^{\sqrt{\log s}/\varepsilon}
    \end{align*}
    at which point, $C_k \leq 0 \leq \varepsilon$, so \topdown returns a tree of size $k+1$.
\end{proof}

\section{Lower bounds on $\TopDownDTsize$ for general functions: Proof of~\Cref{thm:TD-lower}}

\subsection{Size separation for exact representation: Proof of~\Cref{thm:TD-lower}(a)}

We begin with a simple family of functions $\{ f_h \}_{h \in \N}$ whose \textsc{BuildTopDownDT} tree has exponential size compared to the optimal tree.  Each $f_h$ is a function over $3h+1$ boolean variables $x^{(1)}_1,x^{(1)}_2,\ldots,x^{(h)}_1,x^{(h)}_2, y^{(1)},\ldots,y^{(h)},z$, and is defined inductively as follows:   
\[f_0(z) = z, \] 
and for $h\ge 1$,
\[ f_h(x,y,z) = 
\begin{cases}
y^{(h)} & \text{if $x^{(h)}_1 \vee x^{(h)}_2$} \\
f_{h-1}(x,y,z) & \text{otherwise}.
\end{cases}
\]
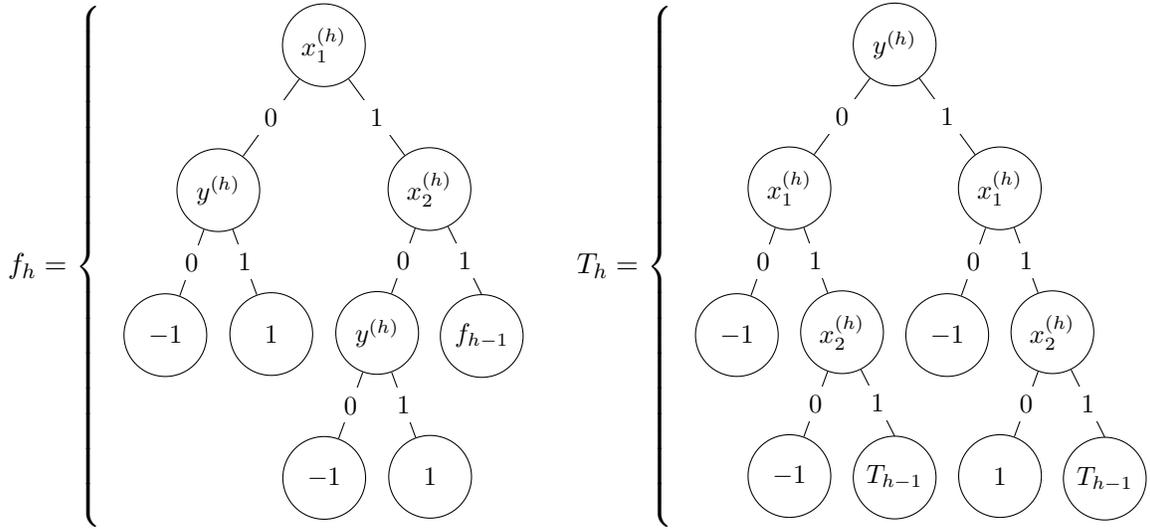
\begin{figure}[tb]
  \captionsetup{width=.9\linewidth}
    \centering
    \forestset{
  triangle/.style={
    node format={
      \noexpand\node [
      draw,
      shape=regular polygon,
      regular polygon sides=3,
      inner sep=0pt,
      outer sep=0pt,
      \forestoption{node options},
      anchor=\forestoption{anchor}
      ]
      (\forestoption{name}) {\foresteoption{content format}};
    },
    child anchor=parent,
  }
 }

\begin{align*}
    &f_h = 
    \left\{\hspace{3mm}{\small
    \!\begin{gathered}
    \begin{forest}
    for tree={
        grow=south,
        circle, draw, minimum size=11mm, l sep = 8mm, s sep = 3mm, inner sep = 1mm
    }
    [$x_1^{(h)}$
        [$y^{(h)}$, edge label = {node [midway, fill=white] {$0$} }
            [$-1$, edge label = {node [midway, fill=white] {$0$} }]
            [$1$, , edge label = {node [midway, fill=white] {$1$} }]
        ]
        [$x_2^{(h)}$, edge label = {node [midway, fill=white] {$1$} }
            [$y^{(h)}$, edge label = {node [midway, fill=white] {$0$} }
                [$-1$, edge label = {node [midway, fill=white] {$0$} }]
                [$1$, edge label = {node [midway, fill=white] {$1$} }]
            ]
            [$f_{h-1}$, triangle, edge label = {node [midway, fill=white] {$1$} }]
        ]
    ]
    \end{forest}
    \end{gathered}}\right.
    &T_h = \left\{\hspace{3mm}\small{
    \!\begin{gathered}
    \begin{forest}
    for tree={
        grow=south,
        circle, draw, minimum size=11mm, l sep = 8mm, s sep = 3mm, inner sep = 1mm
    }
    [$y^{(h)}$
        [$x_1^{(h)}$, edge label = {node [midway, fill=white] {$0$} }
            [$-1$, edge label = {node [midway, fill=white] {$0$} }]
            [$x_2^{(h)}$, , edge label = {node [midway, fill=white] {$1$} }
                [$-1$, edge label = {node [midway, fill=white] {$0$} }]
                [$T_{h-1}$, triangle, edge label = {node [midway, fill=white] {$1$} }]
            ]
        ]
        [$x_1^{(h)}$, edge label = {node [midway, fill=white] {$1$} }
            [$-1$, edge label = {node [midway, fill=white] {$0$} }]
            [$x_2^{(h)}$, edge label = {node [midway, fill=white] {$1$} }
                [$1$, edge label = {node [midway, fill=white] {$0$} }]
                [$T_{h-1}$, triangle, edge label = {node [midway, fill=white] {$1$} }]
            ]
        ]
    ]
    \end{forest}
    \end{gathered}}\right.
\end{align*}
    \caption{Diagrams exhibiting a function with exponential difference between the optimal decision tree size and \TopDownDTsize. The left diagram shows how to compute $f_h$ with a decision tree of size $O(h)$. The right diagram shows $T_h$, the tree \topdown{} builds, which has size $2^{\Omega(h)}$.}
    \label{fig:Non Monotone Exact Tree}
\end{figure}

\pparagraph{The structure of $\BuildTopDownDT(f_h)$.} We see that $y^{(h)}$ has influence $\frac{3}{4}$, both $x^{(h)}_1$ and $x^{(h)}_2$ have influence $\frac{1}{4}$, and each variable in $f_{h-1}$ has influence $< \frac{1}{4}$. $\BuildTopDownDT(f_h)$ therefore queries $y_k$ at the root. In the restrictions of $f_h$ obtained by setting $y^{(h)}$ to a constant, $x^{(h)}_1$ and $x^{(h)}_2$ have equal influence of $\frac1{4}$ and each variable in $f_{h-1}$ has influence $< \frac1{4}$. By setting either $x^{h)}_1$ or $x^{(h)}_2$ to a constant, we get a subfunction where the other $x^{(h)}$-variable has influence $\frac1{2}$ and each node in $f_{h-1}$ has influence $< \frac1{2}$. Thus, $\BuildTopDownDT(f_h)$ builds the tree $T_h$ depicted in~\Cref{fig:Non Monotone Exact Tree}.  

We see that each $T_h$ contains two copies of $T_{h-1}$. It follows that the optimal size of $f_h$ is $O(h)$, whereas the size of $T_h$ is $2^{\Omega(h)}$: a size separation of $\TopDownDTsize(f_h) = 2^{\Omega(s)}$ where $s$ denotes the optimal size of $f_h$.

\subsection{Size separation for approximate representation: Proof of~\Cref{thm:TD-lower}(b)}

\noindent {\bf Warmup/intuition: An $s$ versus $s^{\Omega(\log(1/\eps))}$ separation.} Before proving~\Cref{thm:TD-lower}(b), we first give a brief, informal description of how  a simple modification to the family of functions $\{f_h\}_{h\in \N}$ in~\Cref{thm:TD-lower}(a) above yields a separation of $\TopDownDTsize(f,\eps) = s^{\Omega(\log(1/\eps))}$ for approximate representation.~\Cref{thm:TD-lower}(b)---which improves this to a superpolynomial separation even for constant $\eps$---builds on these ideas, but the family of functions and the proof of the lower bound are significantly more involved. 

Consider replacing each $y^{(h)}$ variable in the definition of $f_h$ with the parity of $k$ variables $y^{(h)}_1 \oplus \cdots \oplus y^{(h)}_k$, i.e.~consider the following variant $\tilde{f}_h$ of $f_h$: 
\[ \tilde{f}_h(x,y,z) = 
\begin{cases}
y^{(h)}_1 \oplus \cdots \oplus y^{(h)}_k & \text{if $x^{(h)}_1 \vee x^{(h)}_2$} \\
\tilde{f}_{h-1}(x,y,z) & \text{otherwise}.
\end{cases}
\]
Just like the single $y^{(h)}$ variable in $f_h$, we see that the $k$ many $y^{(h)}_i$ variables are the most influential in $\tilde{f}_h$ (each having influence $\frac{3}{4}$).  Furthermore, each $y^{(h)}_i$ variable remains the most influential even under \emph{any} restriction to \emph{any} number of the other $y^{(h)}_j$ variables.  Therefore the tree $\tilde{T}_h$ that $\BuildTopDownDT$ builds for $\tilde{f}_h$ first queries all $k$ many $y^{(h)}$ variables.  At each of the $2^k$ resulting leaves, $x^{(h)}_1$ and $x^{(h)}_2$ are then queried, followed by a copy of $\tilde{T}_{h-1}$, the tree that $\BuildTopDownDT$ recursively constructs for $\tilde{f}_h$, in the branch corresponding to $x^{(h)}_1 = x^{(h)}_2 = 1$.  The fact that there are $2^k$ copies of $\tilde{T}_{h-1}$ within $\tilde{T}_{h}$ should be contrasted with the fact that the tree $T_h$ in~\Cref{thm:TD-lower}(a) contains just  two copies of $T_{h-1}$; recall~\Cref{fig:Non Monotone Exact Tree}. 

It is straightforward to see that there is a tree of size $O(h\cdot 2^k)$ that computes $\tilde{f}_h$.  This tree is built by first querying the $x^{(h)}$ variables before the $y^{(h)}$ variables, and recursing on just one of the $\Omega(2^k)$ many resulting leaves.  On the other hand, by first querying the $y^{(h)}$ variables followed by the $x^{(h)}$ variables, $\BuildTopDownDT$ recurses on $\Omega(2^k)$ many branches while only correctly classifying a $\frac{3}{4}$ fraction of inputs.  Choosing $h = \Theta(\log(1/\eps))$, we get a separation of $O(h\cdot 2^k)$ versus $2^{\Omega(kh)}$, or equivalently, $s$ versus $s^{\Omega(\log(1/\eps))}$.

\subsubsection{Proof of~\Cref{thm:TD-lower}(b)}

Before defining the family of functions witnessing the separation, we define a couple of basic boolean functions and state a few of  their properties that will be useful for our analyses:  

\begin{definition}[$\Tribes$]
    \label{def:Tribes} For any input length $r$, let $w$ be the largest integer such that $(1-2^{-w})^{r/w} \le \frac1{2}$.  The $\Tribes_r : \zo^r \to \{ \pm 1\}$ function is defined to be the function computed by the read-once DNF with $\lfloor \frac{r}{w}\rfloor$ terms (over disjoint sets of variables) of width exactly $w$: 
    \[ \Tribes_r(z) = (z_{1,1}\wedge \cdots \wedge z_{1,w}) \vee \cdots \vee (z_{t,1} \wedge \cdots \wedge z_{t,w}) \quad \text{where $t \coloneqq \lfloor \lfrac{r}{w}\rfloor,$} \]
    and where we adopt the convention that $-1$ represents logical \textsc{False} and $1$ represents logical \textsc{True}.
    \end{definition}

The following facts about the $\Tribes$ function are standard  (see Chapter \S4.2 of~\cite{ODBook}) and can be easily verified:

\begin{fact}[Properties of $\Tribes_r$]\label{fact:Tribes-properties} \ 
\begin{itemize} 
\item[$\circ$] $\Pr[\Tribes_r(\bz) = 1] = \frac1{2} -O\big(\frac{\log r}{r}\big).$
\item[$\circ$] $\Inf(\Tribes_r) = (1\pm o(1))\cdot \ln r$ and consequently, $\Inf_i(\Tribes_r) = (1\pm o(1))\cdot \frac{\ln r}{r}$ for all $i\in [n]$.
\item[$\circ$] $w = \log r - \log \ln r \pm O(1)$.
\item[$\circ$] $\size(\Tribes_r) \le w^{O(r/w)} = 2^{O(r\log\log r/\log r)}$.
\end{itemize} 

\end{fact}

\begin{definition}[$\Threshold$]
    \label{def:Threshold} 
    For any input length $\ell$ and $t \in \{0,1,\ldots,\ell\}$, the $\Threshold_{\ell,t} : \zo^\ell \to \{ \pm 1\}$ function is defined to be 
    \[ \Threshold_{\ell,t}(x) = 1 \Longleftrightarrow \sum_{i=1}^\ell x_i \le t. \] 
\end{definition}

\pparagraph{Defining the family of functions witnessing the separation.} Consider the following family of functions $\{ f_h \}_{h \in \N}$.  Each $f_h$ is a function over $h(\ell + k)+r$ boolean variables $x^{(1)},x^{(2)},\ldots,x^{(h)} \in \zo^\ell$,$ y^{(1)},\ldots,y^{(h)} \in \zo^k$, and $z\in \zo^r$, and is defined inductively as follows:   

\[f_0(z) = \Tribes_r(z), \] 
and for $h\ge 1$,
\[ f_h(x,y,z) = 
\begin{cases}
\Parity_k(y^{(h)}) & \text{if $\Threshold_{\ell,1}(x^{(h)}) = 1$} \\
f_{h-1}(x,y,z) & \text{otherwise}.
\end{cases}
\]

\begin{claim}[Optimal decision tree size of $f_h$]
\label{claim:opt-size-general-approx} 
\begin{align*} \size(f_h) &\le \ell^{O(h)}\cdot (\size(\Parity_k) + \size(\Tribes_r)) \\
&\le \ell^{O(h)} \cdot (2^k + 2^{O(r\log\log r/\log r)}). 
\end{align*} 
\end{claim}

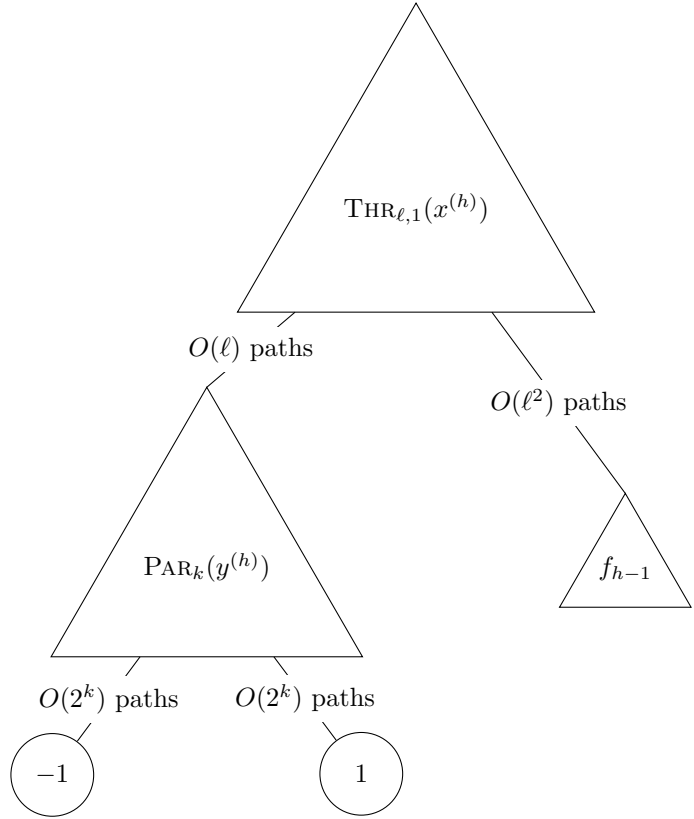
\begin{figure}[h]
  \captionsetup{width=.9\linewidth}
    \centering
    \forestset{
  triangle/.style={
    node format={
      \noexpand\node [
      draw,
      shape=regular polygon,
      regular polygon sides=3,
      inner sep=0pt,
      outer sep=0pt,
      \forestoption{node options},
      anchor=\forestoption{anchor}
      ]
      (\forestoption{name}) {\foresteoption{content format}};
    },
    child anchor=parent,
  }
 }

\[ {\small
\!\begin{gathered}
\begin{forest}
for tree={
    grow=south,
    minimum size=11mm, l sep = 10mm, s sep = 30mm
}
[$\textsc{Thr}_{\ell,1}(x^{(h)})$, triangle, draw
    [$\textsc{Par}_k(y^{(h)})$, triangle, edge label = {node[midway, fill=white] {$O(\ell)$ paths}}
        [$-1$, circle, draw, edge label = {node[midway, fill=white] {$O(2^k)$ paths}}]
        [$1$, circle, draw, edge label = {node[midway, fill=white] {$O(2^k)$ paths}}]
    ]
    [$f_{h-1}$, triangle, edge label = {node[midway, fill=white] {$O(\ell^2)$ paths}}]
  ]
]    
\end{forest}
\end{gathered} }
\]
    \caption{A small decision tree for $f_h$.}
    \label{fig:Non Monotone Approx Optimal Tree}
\end{figure}

\begin{proof}
Please refer to figure~\Cref{fig:Non Monotone Approx Optimal Tree}.  We first build a tree of size $O(\ell^2)$ that evaluates $\Threshold_{\ell,1}(x^{(h)})$.  Of these leaves, $\ell + 1$ descend into a tree computing $\Parity_k(y^{(h)})$, which has size $2^k$. The others descend into a tree computing $f_{h-1}$. This yields the recurrence 
\begin{align*}
    \size(f_h) &\le O(\ell) \cdot \size(\Parity_k) + O(\ell^2) \cdot \size(f_{h-1}) \\
    &\le O(\ell \cdot 2^k) + O(\ell^2) \cdot \size(f_{h-1}) \\
    \size(f_0) &= \size(\Tribes_r) \le 2^{O(r\log\log r/\log r)}, \tag*{(Recall \Cref{fact:Tribes-properties})}
\end{align*}
and the claim follows. 
\end{proof}

The remainder of this section will be devoted to proving a lower bound on $\TopDownDTsize(f,\eps)$. \Cref{fig:Non Monotone Approx TD Tree} should be contrasted with~\Cref{fig:Non Monotone Approx Optimal Tree}. 

\begin{figure}[tb]
  \captionsetup{width=.9\linewidth}
    \centering
    \forestset{
  triangle/.style={
    node format={
      \noexpand\node [
      draw,
      shape=regular polygon,
      regular polygon sides=3,
      inner sep=0pt,
      outer sep=0pt,
      \forestoption{node options},
      anchor=\forestoption{anchor}
      ]
      (\forestoption{name}) {\foresteoption{content format}};
    },
    child anchor=parent,
  }
 }

\[ {\small
\!
\begin{gathered}
\begin{forest}
for tree={
    grow=south,
    minimum size=11mm, l sep = 10mm, s sep = 30mm
}
[$\textsc{Par}_k(y^{(1)})$, triangle, draw
    [$\textsc{Thr}_{\ell,1}(x^{(1)})$, triangle, edge label = {node[midway, fill=white] {$O(2^k)$ paths}}
        [$-1$, circle, draw, edge label = {node[midway, fill=white] {$O(\ell)$ paths}}]
        [$~T_0~$, triangle, draw, edge label = {node[midway, fill=white] {$O(\ell^2)$ paths}}]
    ]
    [$\textsc{Thr}_{\ell,1}(x^{(1)})$, triangle, edge label = {node[midway, fill=white] {$O(2^k)$ paths}}
        [$1$, circle, draw, edge label = {node[midway, fill=white] {$O(\ell)$ paths}}]
        [$~T_0~$, triangle, draw, edge label = {node[midway, fill=white] {$O(\ell^2)$ paths}}]
    ]
  ]
]    
\end{forest}
\end{gathered} }
\]
    \caption{The tree $T_1$ that $\topdown{}$ builds for $f_1$. Since $y^{(1)}$ has all the most influential variables, $\topdown{}$ puts them all at the root. As a result, it ends up building a significantly larger tree than optimal (cf.~\Cref{fig:Non Monotone Approx Optimal Tree}). Notice that the size of $T_1$ is $\Omega(2^k)$ times as large as $T_0$. This leads to exponential growth of the tree size as a function of $h$.}
    \label{fig:Non Monotone Approx TD Tree}
\end{figure}
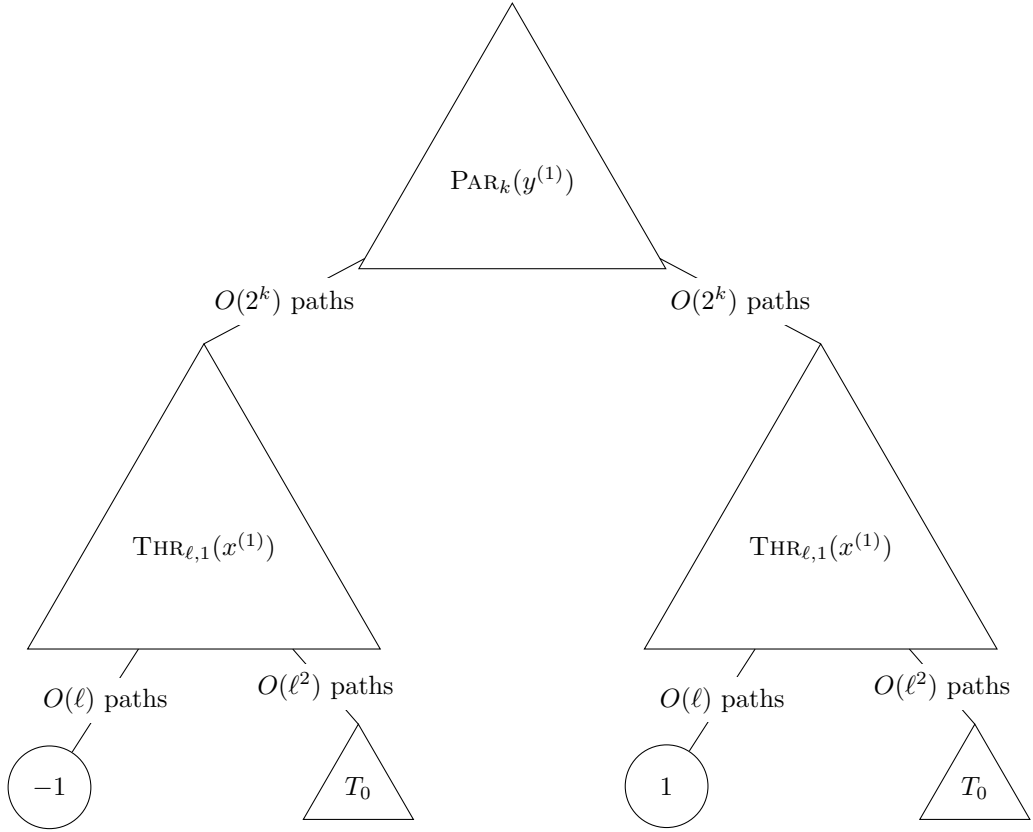

\pparagraph{The structure of $\BuildTopDownDT(f_h)$} The following helper lemma will be useful in determining the structure of the tree \topdown{} produces.

\begin{lemma}[Preservation of influence order]
    \label{lemma: subfunction most influential}
    Let $f:\{0,1\}^{\overline{S}} \times \{0,1\}^{S} \rightarrow \{\pm 1\}$ and $\tilde{f} : \zo^{S} \to \{\pm 1\}$ be two functions satisfying the following: there is a function $g:\{0,1\}^{\overline{S}} \times \{\pm 1\} \rightarrow \{\pm 1\}$ such that: 
    \begin{equation*}
        f(a,b) = g(a, \tilde{f}(b)) \quad \text{for all $a \in \zo^{\overline{S}}$, $b\in \zo^S$.} 
    \end{equation*}
    Then for all variables $v_1,v_2 \in S$, \[ \Inf_{v_1}(\tilde{f}) \ge \Inf_{v_2}(\tilde{f})\quad \text{if and only if}\quad \Inf_{v_1}(f) \ge \Inf_{v_2}(f). \] 
\end{lemma}

\begin{proof} 
This holds by noting that for $v \in \{v_1,v_2\}$, 
\begin{align*} 
\Inf_{v}(f) &= \Prx_{\ba,\bb}[f(\ba,\bb) \ne f(\ba,\bb^{\oplus v})] \\
&= \Prx_{\ba,\bb}[g(\ba,\tilde{f}(\bb))\ne g(\ba,\tilde{f}(\bb^{\oplus v}))] \\
&= \Prx_{\bb}[\tilde{f}(\bb) \ne \tilde{f}(\bb^{\oplus v})] \cdot \Prx_{\ba}[g(\ba,-1)\ne g(\ba,1)] \\
&= \Inf_v(\tilde{f}) \cdot \Prx_{\ba}[g(\ba,-1)\ne g(\ba,1)].
\end{align*} 
The lemma follows since $\displaystyle \Prx_{\ba}[g(\ba,-1)\ne g(\ba,1)]$ does not depend on $v$ (and hence is the same regardless of whether $v = v_1$ or $v = v_2$).
\end{proof}

\Cref{lemma: subfunction most influential} is especially well-suited for our inductively-defined family of functions $\{ f_h \}_{h\in \N}$.  For each $i \in \{0,1,\ldots,h\}$, we let $S_i$ denote the relevant variables of $f_i$.  Therefore 
\begin{align*}
S_0 &= \{z_1,\ldots,z_r\} \\
S_{i+1} &= S_i \sqcup \{x^{(i)}_1,\ldots,x^{(i)}_\ell, y^{(i)}_1,\ldots, y^{(i)}_k\} 
\end{align*} 

\begin{observation} {\it 
For all $i \in \{0,1,\ldots,h\}$, there exists $g_i$ such that 
\begin{equation} 
f_h(a,b)  = g_i(a,h_i(b)) \quad \text{for all $a\in \zo^{S_h\setminus S_i}$ and $b \in \zo^{S_i}$.} \label{eq:nested}
\end{equation} 
Consequently, we may apply~\Cref{lemma: subfunction most influential} to get that for all $v_1,v_2 \in S_i$, we have that 
\[  \Inf_{v_1}(f_i) \ge \Inf_{v_2}(f_i) \quad \text{if and only if}\quad  \Inf_{v_1}(f_h) \ge \Inf_{v_2}(f_h).  \]} 
\end{observation} 

We note the following corollary, which is a straightforward consequence of the observation that the property (\ref{eq:nested}) is preserved under restrictions: 

\begin{corollary}[Preservation of influence order under restrictions] 
\label{cor: subfunction most influential} 
Let $\pi$ be any restriction.  For all $i\in \{0,1,\ldots,h\}$, we have that 
\[ \Inf_{v_1}((f_i)_\pi) \ge \Inf_{v_2}((f_i)_\pi) \quad \text{if and only if}\quad  \Inf_{v_1}((f_h)_\pi) \ge \Inf_{v_2}((f_h)_\pi).  \] 
\end{corollary}

\pparagraph{Lower bounding the size of $\BuildTopDownDT(f,\eps)$.}  Let $T_\exact$ denote the tree returned by $\BuildTopDownDT(f_h)$, and $T_\mathrm{approx}$ denote the tree returned by $\BuildTopDownDT(f_h,\eps)$.  (So $T_\exact$ computes $f_h$, and $T_{\mathrm{approx}}$ is an $\eps$-approximation of $f_h$.)  Our goal is to lower bound the size of $T_{\mathrm{approx}}$.  We will in fact establish something stronger: our lower bound holds for\emph{any pruning} of $T_\exact$ that is an $\eps$-approximation of $T_\exact$, where a pruning of a tree $T$ is any tree obtained by iteratively removing leaves from $T$ in a bottom-up fashion.  Since $\BuildTopDownDT(f,\eps)$ is simply $\BuildTopDownDT(f_h)$ terminated early, we have that $T_{\mathrm{approx}}$ is indeed a pruning of $T_\exact$.

Let $V_\exact$ be defined as follows: 
\[ V_\exact \coloneqq \{ v \colon \text{$v$ is the first node in a path of $T_\exact$ that queries a $z$-variable} \}. \] 
We define $V_{\mathrm{approx}} \sse V_\exact$ analogously. At a very high level, our proof of~\Cref{thm:TD-lower}(b) will proceed by showing that $V_\exact$ has large size, and $V_\mathrm{approx}$ has to contain many nodes in $V_\exact$. For the remainder of this proof, we will need that $r$ and $\ell$ are chosen to satisfy: 
\begin{equation} 
\frac{2\ln r}{r} < 2^{-\ell}.   \label{eq:tribes small influence}
\end{equation}

\begin{lemma}[All nodes in $V_\exact$ occur deep within $T_\exact$]
\label{lem:long path} 
    Fix $v \in V_\exact$ and let $\pi$ denote the path in $T_\exact$ that leads to $v$.  Then $|\pi| \ge kh$.  
\end{lemma}

\begin{proof}
Suppose without loss of generality that $v$ is a query to $z_1$.  We claim that $y^{(i)}_j \in \pi$ for all $i \in [h]$ and $j\in [k]$, from which the lemma follows.  Fix $i \in [h]$.  We will prove there are at least $k$ queries to variables in $S_i$ within $\pi$, and that the first $k$ of these queries have to be $y^{(i)}_j$ for $j \in [k]$.  We prove both these claims simultaneously by induction on $k$.    
\begin{itemize}[leftmargin=0.5cm]
\item[$\circ$] (Base case.)  Seeking a contradiction, suppose $\pi$ does not contain any queries to variables in $S_i$, in which case $(f_i)_\pi \equiv f_i$.  Since $z_1$ is the variable queried at the root of $(f_h)_\pi$, it is the most influential variable within $(f_h)_\pi$.  By~\Cref{cor: subfunction most influential}, it follows that $z$ is the most influential variable within $(f_i)_\pi \equiv f_i$.  This contradicts~\Cref{eq:tribes small influence} since 
\[ \Inf_{y^{(i)}_1}(f_i) = \frac{\ell+1}{2^\ell} \quad \text{and} \quad \Inf_{z_1}(f_i) < \Inf_{z_1}(\Tribes_r) = (1\pm o(1)) \cdot \frac{\ln r}{r}.\] 
Therefore $\pi$ has to contain at least one variable in $S_i$.  Let $u \in \pi$ be the first query to a variable in $S_i$, which we claim must be $y^{(i)}_j$ for some $j$.  Let $\pi_u \subset \pi$ be the path in $T_\exact$ that leads to $u$.  Again, we have that $u$ must be the most influential variable within $(f_h)_{\pi_u}$, and hence, by~\Cref{cor: subfunction most influential}, it is the most influential within $(f_i)_{\pi_u}$.  Since $\pi_u$ does not contain any queries to variables in $S_i$, we have that $(f_i)_{\pi_u} \equiv f_i$, and hence $u$ must be $y^{(i)}_j$ for some $j$ since these are the most influential variables within $f_i$. 
\item[$\circ$] (Inductive step.)  Fix $k' < k$, and suppose we have established that there are at least $k'$ queries to variables in $S_i$ within $\pi$, the first $k'$ of which are to $y^{(i)}$-variables.  We first claim that there is at least one more query to variable in $S_i$ within $\pi$.  Suppose not.  It follows that $z_1$ must be the most influential variable within $(f_h)_\pi$, and hence, by~\Cref{cor: subfunction most influential}, it is the most influential variable within $(f_i)_\pi$.  This is a contradiction, since $z_1$ is less influential than any of the $k-k'$ many $y^{(i)}_j$ variables that are not queried by $\pi$. 

Therefore $\pi$ has to contain at least one more query to a variable in $S_i$.  Let $u \in \pi$ be the $(k+1)^{st}$ query to a variable in $S_i$, which we claim must be $y^{(i)}_j$ for some $j$.  Let $\pi_u \subset \pi$ be the path in $T_\exact$ that leads to $u$.  Again, we have that $u$ must be the most influential variable within $(f_h)_{\pi_u}$, and hence, by~\Cref{cor: subfunction most influential}, it is the most influential within $(f_i)_{\pi_u}$.  Since $\pi_u$ contains exactly $k'$ queries to variables $S_i$, and all these queries are to $y^{(i)}$ variables, we have that $u$ must be $y^{(i)}_j$ for one of the remaining $k-k'$ many $y^{(i)}$-variables since these are the most influential variables within $(f_i)_{\pi_u}$. 
\end{itemize} 
This completes the inductive proof of~\Cref{lem:long path}.  \qedhere

\end{proof}

\begin{lemma}
\label{lem:err on half} 
    Fix $v \in V_\exact$ and let $\pi$ denote that path in $T_\exact$ that leads to $v$.  Then $(f_h)_{\pi} \equiv \Tribes_r$.  
\end{lemma}

\begin{proof}
Suppose $(f_h)_{\pi} \not\equiv \Tribes_r$. Our proof of~\Cref{lem:long path} shows that $\pi$ contains every $y$-variable, so it must be the case that some $x$-variable remains relevant (i.e.~has nonzero influence) in $(f_h)_{\pi}$. Let $i^* \geq 1$ be the highest value of $i$ for which there is a relevant $x^{(i)}$-variable in $(f_h)_{\pi}$.  Assume without loss of generality that $x^{(i)}_1$ remains relevant, and that $z_1$ is that $z$-variable that is queried at $v$.   

Since $z_1$ is queried at the root of $(f_h)_\pi$, we have that it must be maximally influential in $(f_h)_\pi$, and in particular, 
\[ \Inf_{z_1}((f_h)_\pi) \ge \Inf_{x^{(i^*)}_1}((f_h)_\pi). \]  
Applying~\Cref{cor: subfunction most influential}, we infer that 
\begin{equation}  \Inf_{z_1}((f_{i^*})_\pi) \ge \Inf_{x^{(i^*)}_1}((f_{i^*})_\pi).\label{eq:blah}
\end{equation} 
Let us say that an input $(x,y,z)$ to $f_{i^*}$ is \emph{$z$-dependent} if
\[ \Threshold_{\ell,1}(x^{(i)}) = 0 \quad \text{for all $1 \le i \le i^*$}. \] 
Note that the output of $f_{i^*}$ on any $z$-dependent input is $\Tribes_r(z)$.  Since $\pi$ contains every $y$-variable, it fixes $\Parity_k(y^{(i^*)})$ to either $-1$ or $1$; we assume without loss of generality that $\Parity_k(y^{(i^*)})_\pi \equiv 1$.   We have that 
\begin{align*} 
\Inf_{x_1^{(i^*)}}((f_{i^*})_{\pi}) &\geq \Prx_{(\bx,\by,\bz)}[\text{$(\bx,\by,\bz)$ is $z$-dependent}] \cdot  \Prx_{\bz}[\Tribes_r(\bz) \ne 1] \\ 
& \quad \quad  \times  \Inf_{x_1^{(i^*)}}\big(\Threshold_{\ell,1}(x^{(i^*)})_\pi\big) \\
&\ge \Prx_{(\bx,\by,\bz)}[\text{$(\bx,\by,\bz)$ is $z$-dependent}] \cdot  \lfrac1{2}  \cdot 2^{-(\ell-1)} \\
&= \Prx_{(\bx,\by,\bz)}[\text{$(\bx,\by,\bz)$ is $z$-dependent}] \cdot 2^{-\ell}.
\end{align*}
The second inequality uses the fact that $\Inf_{x_1^{(i^*)}}(\Threshold_{\ell,1}(x^{(i^*)})_\pi) \ge 2^{-(\ell-1)}$, which holds with equality when exactly one other $x^{(i^*)}$-variable is in $\pi$ and that variable is set to $1$.\footnote{In this derivation, we have assumed that $\Tribes_r$ is {\sl perfectly} balanced, i.e.~that $\Pr[\Tribes_r(\bz) = 1] = \frac1{2}$, when in fact $\Pr[\Tribes_r(\bz)=1] = \frac1{2} \pm o(1)$ (recall~\Cref{fact:Tribes-properties}).  The same proof goes through if one carries around the additive $o(1)$ factor.}

On the other hand, we have that 
\begin{align*} 
\Inf_{z_1}((f_{i^*})_{\pi}) &\leq \Prx_{(\bx,\by,\bz)}[\text{$(\bx,\by,\bz)$ is $z$-dependent}]  \cdot  \Inf_{z_1}[\Tribes_r(z)] \\
&< \Prx_{(\bx,\by,\bz)}[\text{$(\bx,\by,\bz)$ is $z$-dependent}] \cdot \frac{2\ln r}{r}.  \tag*{(\Cref{fact:Tribes-properties})}
\end{align*}
These bounds on influences, along with~\Cref{eq:blah}, imply that $2^{-\ell} < \frac{2\ln r}{r}$.  This contradicts our assumption on the relationship between $\ell$ and $r$ (\Cref{eq:tribes small influence}), and the proof is complete. 
\end{proof}

We are now ready to lower bound the size of $T_{\mathrm{approx}}$.

\begin{claim}[Lower bound on the size of $T_{\mathrm{approx}}$]
\label{claim:approx-LB}
Fix $\eps \in (0,\frac1{2})$ and let $c = (\frac1{2}-\eps)/2$.  If 
\begin{align}
\left(1-\frac{\ell+1}{2^\ell}\right)^h &\ge (2+c)\eps, \label{eq:many reach Tribes}
\end{align} 
then $|V_\mathrm{approx}| \ge \Omega(\eps \cdot 2^{kh})$.  Consequently, the size of $T_{\mathrm{approx}}$ is also at least $\Omega(\eps \cdot 2^{kh})$. 
\end{claim}

\begin{proof}
An input to $f_h$ reaches \emph{some} node in $V_\exact$ if and only if $\Threshold_{\ell,1}(x^{(i)}) = 0$ for all $1\le i \le h$.  The fraction of inputs that satisfies this is exactly $(1-\frac{\ell+1}{2^\ell})^h$, which is at least $(2+c)\eps$ by our choice of parameters given by~\Cref{eq:many reach Tribes}.  

Fix $v\in V_\exact$. If $v \notin V_\mathrm{approx}$, then $T_{\mathrm{approx}}$ assigns all inputs reaching $v$ the same $-1$ or $+1$ value, whereas $f_h$ labels half of them $-1$ and half of them $+1$ (\Cref{lem:err on half}).  Therefore, $T_\mathrm{approx}$ errors on half of the inputs that each $v$.    On the other hand, if $v \in V_\mathrm{approx}$, we have by~\Cref{lem:long path} that at most a $2^{-kh}$ fraction of inputs reach this specific $v$.   Combining all of the above observations, it follows that 
\[ \error(T_\exact,T_{\mathrm{approx}}) \ge \lfrac{1}{2} \left((2 + c)\eps - |V_{\mathrm{approx}}|\cdot 2^{-kh}  \right).  \] 
Since $\error(T_\exact,T_{\mathrm{approx}}) \le \eps$, it follows that 
\[ \eps \ge \lfrac{1}{2} \left((2 + c)\eps - |V_{\mathrm{approx}}|\cdot 2^{-kh}  \right), \]
and the claim follows by rearranging.
\end{proof}

\Cref{thm:TD-lower}(b) now follows from \Cref{claim:opt-size-general-approx} and~\Cref{claim:approx-LB} by setting parameters appropriately: 

\begin{proof}[Proof of~\Cref{thm:TD-lower}(b)] 
Choosing 
\begin{align*}
h &= \Theta\hspace*{-3pt}\left(\frac{2^{\ell}}{\ell}\cdot \log(1/\eps)\right) \tag*{(to satisfy~\Cref{eq:many reach Tribes})} \\
r &= \Theta(\ell 2^\ell) \tag*{(to satisfy~\Cref{eq:tribes small influence})} \\ 
k &= \Theta(h\log \ell), 
\end{align*}
we may apply \Cref{claim:opt-size-general-approx} and~\Cref{claim:approx-LB} to get that  
\[ \size(f_h) \le 2^{O(k\log k)} \quad \text{whereas} \quad \TopDownDTsize(f,\eps) \ge 2^{\Omega_\eps(k^2/\log\log k)}. \]
This is a separation of $s$ versus $s^{\tilde{\Omega}(\log s)}$. 
\end{proof}

\begin{remark} 
For our choice of parameters above, we have that $s(n) = \size(f_h) =  2^{\tilde{\Theta}(\sqrt{n})}$, where $n  = h(\ell + k)+r$ is the number of  variables of $f_h$.  A standard padding argument yields the same $s$ versus $s^{\tilde{\Omega}(\log s)}$ separation for any function $s(n) \le 2^{\tilde{O}(\sqrt{n})}$.
\end{remark}

\section{Lower bounds on $\TopDownDTsize$ for monotone functions: Proof of~\Cref{thm:TD-lower-monotone}}

\label{sec:monotone-gain}

\subsection{Size separation for exact representation: Proof of~\Cref{thm:TD-lower-monotone}(a)} 

We will give a family of monotone functions, $\{f_h\}_{h \in \mathbb{N}}$ whose \topdown{} tree has exponential size compared to the optimal tree. First, we define a few terms which will be useful for our monotone constructions.
\begin{definition} [Comparing vectors and upper/lower shadows]
    For any $x,y \in \{0,1\}^n$, we use $x \preceq y$ to represent
    \begin{equation*}
        x \preceq y \iff x_i \leq y_i \text{ for all $i\in [n]$}
    \end{equation*}
    and $\succeq$ is defined similarly. For any vector $x$, the \emph{upper shadow} of $x$ is the set of all vectors $y$ such that $x \preceq y$. Similarly, the \emph{lower shadow} of $x$ is the set of all vectors $y$ such that $x \succeq y$.
\end{definition}
\pparagraph{Defining the family of functions witnessing the separation.} Each $f_h$ in $\{ f_h \}_{h\in \N}$ is a function over $5h+1$ boolean variables $x^{(1)},x^{(2)},\ldots,x^{(h)} \in \zo^4$, $ y^{(1)},\ldots,y^{(h)} \in \zo$, and $z\in \zo$, and is defined inductively as follows:   
\[f_0(z) = z, \] 
and for $h\ge 1$, we fix $x^* \coloneqq (0,0,1,1)$ and define 
\[ f_h(x,y,z) = 
\begin{cases}
f_{h-1}(x,y,z) & \text{if $x^{(h)} = x^*$} \\
+1 & \text{if $x^{(h)} \succeq x^*$ and $x^{(h)} \ne x^*$} \\
-1 & \text{if $x^{(h)} \preceq x^*$ and $x^{(h)} \ne x^*$} \\
y & \text{otherwise.} 
\end{cases}
\] 

It is straightforward to verify that $f_h$ is indeed monotone.  We will show that $f_h$ can be computed by a tree of size $O(h)$, but that \topdown{} produces a tree of size $2^{\Omega(h)}$.   For the first claim, we construct a decision tree for $f_h$ directly from its definition. We start with a complete tree on the $x^{(h)}$ variables---this complete tree has size $2^4$, a constant.  At one of the branches, we recursively build a tree for  $f_{h-1}$; at all the other branches, we build a tree of size $1$ or $2$ computing one of $-1, 1,$ or $y^{(h)}$.  The result is a tree of size $O(h)$. 

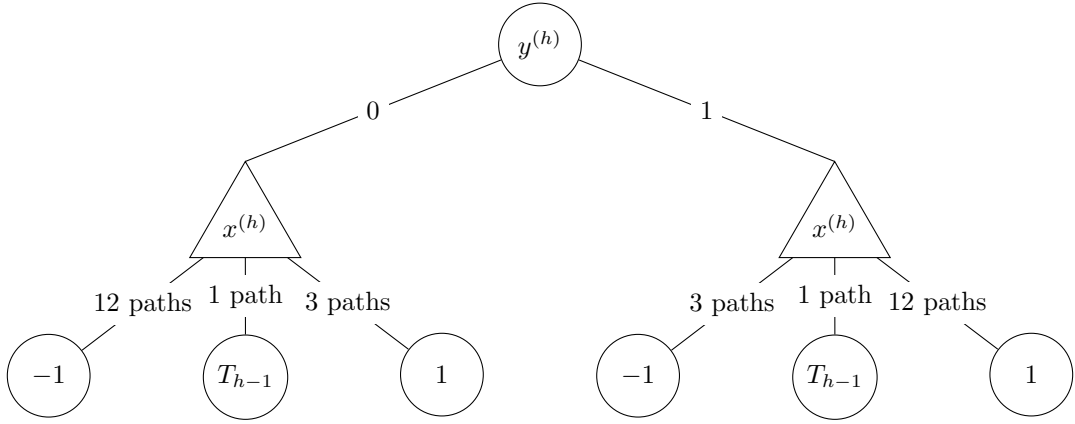
\begin{figure}[tb]
  \captionsetup{width=.9\linewidth}
    \centering
    \forestset{
  triangle/.style={
    node format={
      \noexpand\node [
      draw,
      shape=regular polygon,
      regular polygon sides=3,
      inner sep=0pt,
      outer sep=0pt,
      \forestoption{node options},
      anchor=\forestoption{anchor}
      ]
      (\forestoption{name}) {\foresteoption{content format}};
    },
    child anchor=parent,
  }
 }

\[ {\small
\!\begin{gathered}
\begin{forest}
for tree={
    grow=south,
    minimum size=11mm, l sep = 10mm, s sep = 15mm
}
[$y^{(h)}$, circle, draw
    [$x^{(h)}$, triangle, edge label = {node[midway, fill=white] {$0$}}
        [$-1$, circle, draw, edge label = {node[midway, fill=white] {$12$ paths}}]
        [$T_{h-1}$, circle, draw, edge label = {node[midway, fill=white] {$1$ path}}]
        [$1$, circle, draw, edge label = {node[midway, fill=white] {$3$ paths}}]
    ]
    [$x^{(h)}$, triangle, edge label = {node[midway, fill=white] {$1$}}
        [$-1$, circle, draw, edge label = {node[midway, fill=white] {$3$ paths}}]
        [$T_{h-1}$, circle, draw, edge label = {node[midway, fill=white] {$1$ path}}]
        [$1$, circle, draw, edge label = {node[midway, fill=white] {$12$ paths}}]
    ]
  ]
]    
\end{forest}
\end{gathered} }
\]
    \caption{The tree that \topdown{} builds for $f_h$. It will first query $y^{(h)}$, followed by the variables of $x^{(h)}$. For most choices of $y^{(h)}$ and $x^{(h)}$, the function is determined, and \topdown{} will place a constant leaf equal to $\pm 1$. However, the paths with $y = 0, x^{(h)} = x^*$ and $y=1, x^{(h)} = x^*$ each include a copy of the tree for $T_{h-1}$.}
    \label{fig:Monotone Exact Tree}
\end{figure}

On the other hand, we claim that \topdown{} will build a tree of size $2^{\Omega(h)}$, as depicted in \Cref{fig:Monotone Exact Tree}. In $f_h$, $y^{(h)}$ has influence $\frac{9}{16}$ and all the other variables have influence at most $\frac{1}{2}$. Hence, $y^{(h)}$ will be placed at root. Then, \topdown{} will query enough of $x^{(h)}$ to determine whether the output should be $-1, +1,$ or $f_{h-1}$. If the output should be $f_{h-1}$, which will occur once for each choice of $y$, then the entire tree $T_{h-1}$ will be placed. Hence, the size of $T_h$ is more than double the size of $T_{h-1}$, and \topdown{} builds a tree of size $2^{\Omega(h)}$.

\subsection{Size separation for approximate representation:~\Cref{thm:TD-lower-monotone}(b)} 
\label{section: monotone size seperation}
For any $\varepsilon$, we will prove there exists a function $f$ with optimal tree size $s$ but for which the tree $\topdown(f, \varepsilon)$ builds has size $s^{\tilde{\Omega}(\sqrt[4]{\log s})}$. The following function, a biased version of the $\Tribes$ function defined in~\Cref{def:Tribes}, will be used as a building block in our monotone construction.

\begin{definition}[Biased \textsc{Tribes}]
    Fix any input length $\ell$ and $\delta \in (0,1)$.   We define $\Tribes_{\ell,\delta} : \zo^\ell \to \{ \pm 1\}$ to be the read-once DNF with $\lfloor \frac{\ell}{w}\rfloor$ terms of width exactly $w$ over disjoint sets of variables (with some variables possibly left unused), where $w = w(\ell,\delta) \approx \log(\ell) \pm \log\log(1/\delta)$ is chosen such that $\Prx[\Tribes_{\ell,\delta}(\bx) = 1]$ is as close to $\delta$ as possible.\footnote{Although the acceptance probability of $\Tribes_{\ell,\delta}$ cannot be made {\sl exactly} $\delta$ due to granularity issues, it will be the case that $\Tribes_{\ell,\delta} = \delta \pm o(1)$.  For clarity, we will assume for the rest of this paper that the acceptance probability of $\textsc{Tribes}_\delta$ is exactly $\delta$, noting that all of our proofs still go through if one carries around the $o(1)$ factor.}  
\end{definition}

\begin{fact}[Variable influences in biased $\Tribes$]
\label{fact: biased tribes influence} 
    All variables in $\textsc{Tribes}_{\ell,\delta}$ and $\textsc{Tribes}_{\ell,1-\delta}$ have influence at most 
    \[ (2+o(1))\cdot \delta  \log(1/\delta)\cdot \frac{\log \ell}{\ell}.\] 
\end{fact}
\begin{proof}
We prove the lemma for the case of $\Tribes_{\ell,1-\delta}$. (The calculations for $\Tribes_{\ell,\delta}$ are very similar, and both claims are special cases of more general facts about variable influences in DNF formulas~\cite{ST13}.)  Suppose 
\[ \Tribes_{\ell,1-\delta}(x) = T_1(x) \vee \cdots T_{ \frac{\ell}{w}}(x), \] 
where the $T_i$'s are disjoint terms of width exactly $w$.   We first observe that since 
\begin{align*}
    \delta = \Prx_{\bx \sim \zo^n}[ \Tribes_{\ell,1-\delta}(\bx) =1 ] &= \Pr[\,\text{all $T_i(\bx)$ are falsified by $\bx$}\,] \\
    &= (1-2^{-w})^{\ell/w} \approx e^{-\ell/w2^w},
\end{align*}  
we have that $w = (1\pm o(1))(\log \ell-\log\log\ell - \log\log(1/\delta))$. The influence of any variable $i \in [n]$ on $\Tribes_{\ell,1-\delta}$ is the probability, over a uniform $\bx$ that each other variable $j$ in $i$'s term has $\bx_j = 1$ and all other clauses evaluate to 0 under $\bx$: 
\begin{align*}
\Inf_i(\Tribes_{\ell,1-\delta}) &= 2^{-(w-1)} \cdot (1-2^{-w})^{(\ell/w)-1} \\
&\le 2\delta \cdot 2^{-w} \\
&= (1\pm o(1))\cdot 2\delta\log(1/\delta)\cdot \frac{\log \ell}{\ell}.  \qedhere 
\end{align*}
    \end{proof}

\pparagraph{Defining the family of functions witnessing the separation.} Each $f_h$ in the family $\{ f_h\}_{h\in \N}$ is a function over $h(2\ell + k) + r$ boolean variables $x^{(1,1)},x^{(1,2)},\ldots, x^{(h,1)},x^{(h,2)} \in \zo^\ell$, $y^{(1)},\ldots,y^{(h)}\in \zo^k$, and $z\in \zo^r$, and is defined inductively as follows: 
\[ f_0(z) = \Tribes_r(z), \] 
and for $h\ge 1$, 
\[ f_h(x,y,z) = \begin{cases}
-1 & \text{if $\Tribes_{\ell,\delta}(x^{(h,1)}) = \Tribes_{\ell,1-\delta}(x^{(h,2)}) = 0$} \\
f_{h-1}(x,y,z) & \text{if $\Tribes_{\ell,\delta}(x^{(h,1)}) = 0 \text{ and } \Tribes_{\ell,1-\delta}(x^{(h,2)}) = 1$} \\
\Maj_k(y^{(h)}) & \text{if $\Tribes_{\ell,\delta}(x^{(h,1)}) = 1 \text{ and } \Tribes_{\ell,1-\delta}(x^{(h,2)}) = 0$} \\
+1 & \text{otherwise.} 
\end{cases}
\] 
 Clearly $f_h$ is monotone in $x^{(h,1)}$ and $x^{(h,2)}$. Furthermore, since each of the functions $-1, +1$, and $\Maj_k(y^{(h)}),$ are monotone, if $f_{h-1}$ is monotone then so is $f_h$.

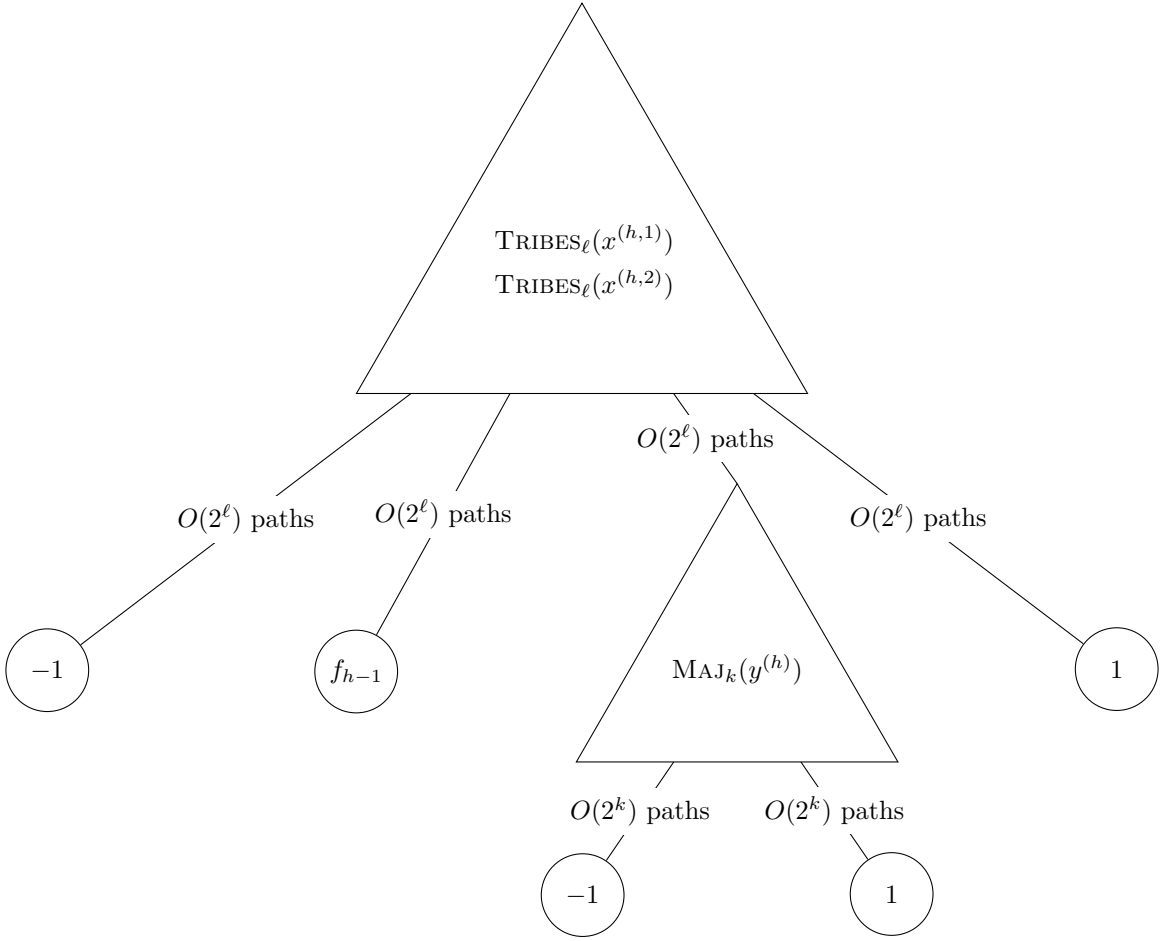
\begin{figure}[tb]
  \captionsetup{width=.9\linewidth}
    \centering
    \forestset{
  triangle/.style={
    node format={
      \noexpand\node [
      draw,
      shape=regular polygon,
      regular polygon sides=3,
      inner sep=0pt,
      outer sep=0pt,
      \forestoption{node options},
      anchor=\forestoption{anchor}
      ]
      (\forestoption{name}) {\foresteoption{content format}};
    },
    child anchor=parent,
  }
 }

\[ {\small
\!\begin{gathered}
{\small
\begin{forest}
for tree={
    grow=south,
    minimum size=11mm, l sep = 12mm, s sep = 30mm
}
 [$\begin{aligned}
    \Tribes_\ell(x^{(h,1)}) \\
    \Tribes_\ell(x^{(h,2)})
    \end{aligned}$, triangle
    [$-1$, circle, draw, edge label = {node[midway, fill=white] {$O(2^\ell)$ paths}}]
    [$f_{h-1}$, circle, draw, edge label = {node[midway, fill=white] {$O(2^\ell)$ paths}}]
    [$\Maj_k(y^{(h)})$, triangle, draw, edge label = {node[midway, fill=white] {$O(2^\ell)$ paths}}
        [$-1$, circle, draw, edge label = {node[midway, fill=white] {$O(2^k)$ paths}}]
        [$1$, circle, draw, edge label = {node[midway, fill=white] {$O(2^k)$ paths}}]
    ]
    [$1$, circle, draw, edge label = {node[midway, fill=white] {$O(2^\ell)$ paths}}]
]
\end{forest}
}
\end{gathered} }
\]
    \caption{A small decision tree that computes $f_h$}
    \label{fig:Monotone approx optimal tree}
\end{figure}

\begin{claim}[Optimal size of $f_h$]
    \label{claim: monotone optimal size}
    Choose any integers $\ell,h,r,k > 0$ and let
    Then, $f_{h,\ell, k,r}$ has optimal decision tree size 
 \begin{align*}
         \size(f_h) &\le (\size(\Tribes_{\ell,\delta})\cdot \size(\Tribes_{\ell,1-\delta}))^{O(h)} \cdot (\size(\Maj_k) + \size(\Tribes_r)) \\
         &\le 2^{O(h\cdot \ell\log\log \ell/\log \ell)}\cdot (2^k + 2^{O(r\log\log r/\log r)}). \tag*{(\Cref{fact:Tribes-properties})}
 \end{align*}
\end{claim}

\begin{proof}
As in the proofs of the previous separations, this upper bound is witnessed by the natural decision tree that one builds by following the definition of $f_h$.  This tree first evaluates $\Tribes_{\ell,\delta}(x^{(h,1)})$ followed by $\Tribes_{\ell,1-\delta}(x^{(h,2)})$, resulting in a tree of size $(\size(\Tribes_{\ell,\delta})\cdot \size(\Tribes_{\ell,1-\delta}))$.  At the end of each branch, we either recursively build a tree for $f_{h-1}$, or a tree for $\Maj_k(y^{(h)})$, or place constants $\{\pm 1\}$ as leaves.  Please refer to~\Cref{fig:Monotone approx optimal tree}. 
\end{proof}

The remainder of this section is devoted to lower bounding $\TopDownDTsize(f_h,\eps)$, the size of the tree $T_{\mathrm{approx}}$ that $\BuildTopDownDT$ constructs to $\eps$-approximate $f_h$.

\subsubsection{``Mostly precedes"}

By choosing parameters appropriately, we will ensure that when $T_{\mathrm{approx}}$ begins by querying the variables of $\Maj_k(y^{(h)})$.  The first technical challenge that arises is the following: unlike $\Parity_k(y^{(h)})$ in our proof of~\Cref{thm:TD-lower}(b), the influence of variables in $\Maj_k(y^{(h)})$ changes as variables are queried. For example, the influence of the remaining variables of $\Maj_k(y^{(h)})$ after $\frac{k}{2}$ variables have been queried is $0$ if all of the queried variables are $1$ and is $\Theta(\frac{1}{\sqrt{k}})$ if half of the variables queries are $0$ and half are $1$. Hence, in $T_{\mathrm{approx}}$, the number of nodes from $y^{(h)}$ queried before some non-$y^{(h)}$-variable is queried will vary by path.  (In other words, the analogue of~\Cref{lem:long path} in the proof of~\Cref{thm:TD-lower}(b) is somewhat trickier to establish.)  

To handle this, we define the following notion, which will allow us to show that {\sl most} $y^{(h)}$-variables are before other variables in {\sl most} paths of the tree (\Cref{cor: whee}). 
\begin{figure}[tb]
    \centering
      \captionsetup{width=.9\linewidth}
    \forestset{
  triangle/.style={
    node format={
      \noexpand\node [
      draw,
      shape=regular polygon,
      regular polygon sides=3,
      inner sep=0pt,
      outer sep=0pt,
      \forestoption{node options},
      anchor=\forestoption{anchor}
      ]
      (\forestoption{name}) {\foresteoption{content format}};
    },
    child anchor=parent,
  }
 }

\[{\small
\!\begin{gathered}
{\small
\begin{forest}
for tree={
    grow=south,
    minimum size=11mm, l sep = 12mm, s sep = 30mm
}
[$\Maj_k(y^{(h)})$, triangle, draw
    [$\begin{aligned}
    \Tribes_{\ell}(x^{(h,1)}) \\
    \Tribes_{\ell}(x^{(h,2)})
    \end{aligned}$, triangle, edge label = {node[midway, fill=white] {$2^{\Omega(k/\log(k))}$ paths}}
        [$-1$, circle, draw, edge label = {node[midway, fill=white] {$O(2^\ell)$ paths}}]
        [$T_{h-1}$, circle, draw, edge label = {node[midway, fill=white] {$O(2^\ell)$ paths}}]
        [$1$, circle, draw, edge label = {node[midway, fill=white] {$O(2^\ell)$ paths}}]
    ]
]
\end{forest}
}
\end{gathered} }
\]
    \caption{With appropriately chosen parameters, the $y^{(h)}$-variables are the most influential in $f_h$, so the tree built by \topdown{} will query them first.  Our analysis shows that this leads to a tree of size $2^{\Omega(kh/\log k)}$ (cf.~\Cref{fig:Monotone approx optimal tree}).}
    \label{fig:Monotone approx TD tree}
\end{figure}
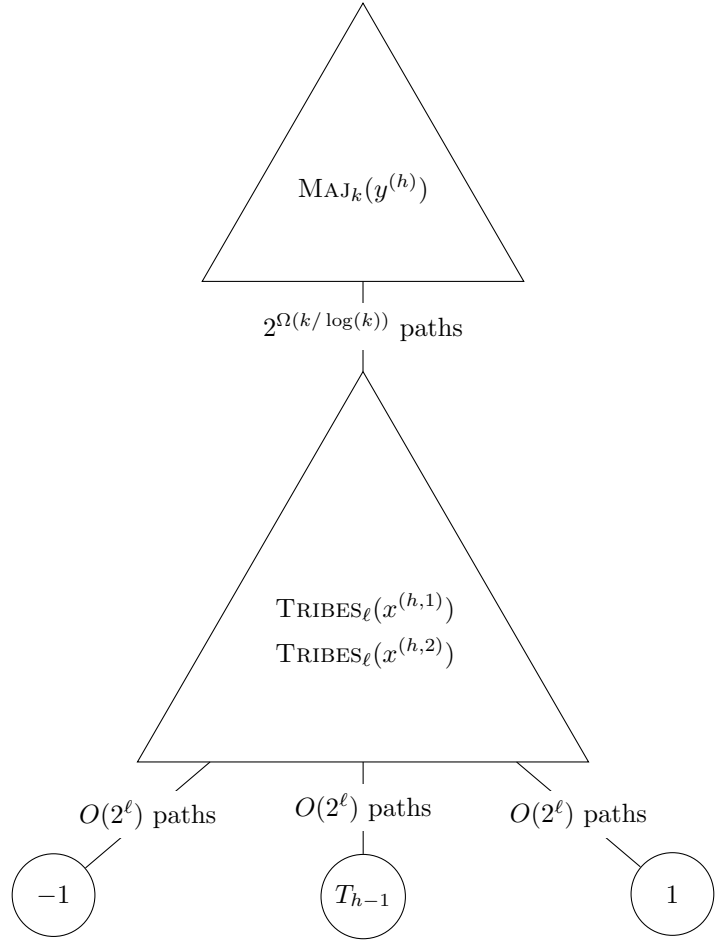

\begin{definition}[Mostly precedes]
\label{def:mostly preceds} 
    Let $S$ be a subset of the relevant variables of $f_h$.  We say that \emph{$y^{(i)}$-variables mostly precede $S$ in $T_{\mathrm{approx}}$} if for every path $\pi$ in $T_{\mathrm{approx}}$ leading to a first query to a variable in~$S$, and every $j\in [k]$, 
    \[ \Inf_{j}\big(\Maj_k(y^{(i)})_\pi\big) \le \frac1{100\sqrt{k}}. \] 
\end{definition}
(For some intuition behind~\Cref{def:mostly preceds}, we note that {\sl pre-restriction}, the influences of variables in $\Maj_k$ are given by: 
\[ \Inf_j(\Maj_k) = \frac1{k} \cdot {k \choose \lfrac{k}{2}} \sim \frac{\sqrt{2/\pi}}{\sqrt{k}} \quad \text{for all $j\in [k]$,}   \]
which is significantly larger than the $\frac1{100\sqrt{k}}$ of~\Cref{def:mostly preceds}.)
With~\Cref{def:mostly preceds} in hand, we now begin to formalize the structure of $T_{\mathrm{approx}}$ as depicted in~\Cref{fig:Monotone approx TD tree}.  For each $i\in [h]$, we define 
\[ R_i = \{  \text{$x^{(i,1)},x^{(i,2)}$, and $z$ variables}\}. 
\]

 \begin{lemma} 
 There is a universal constant $c$ such that the following holds.  Suppose 
  \begin{equation} \frac{c}{\sqrt{k}} \ge \frac1{\delta^2}\cdot  \max\left\{  \frac{\delta \log(1/\delta)\log \ell}{\ell}, \frac{\log r}{r}\right\}.  \label{eq: majority most influential} 
  \end{equation} 
 Then for all $i \in [h]$, we have that $y^{(i)}$-variables mostly precede $R_i$ in $T_{\mathrm{approx}}$.
 \end{lemma}

 \begin{proof} 
 Fix $i\in [h].$ Let $\pi$ be a path in $T_{\mathrm{approx}}$ that leads to a first query to a variable in $v \in R_i$. Since $v$ is maximally influential in $(f_h)_\pi$, we may apply~\Cref{cor: subfunction most influential} to infer that $v$ is also maximally influential in $(f_i)_\pi$ (and in particular, $v$ is more influential than any $y^{(i)}$ variable).  We have that 
 \begin{align*}
     \Inf_v((f_i)_\pi) &\le \max\big\{ \Inf_j(\Tribes_{\ell,\delta}(x^{(i,1)})), \Inf_j(\Tribes_{\ell,1-\delta}(x^{(i,2)})),\Inf_j(\Tribes_r(z))\big\} \\
     &\le \max \left\{(2+o(1))\cdot \delta \log(1/\delta) \cdot \frac{\log \ell}{\ell}, (1+o(1))\cdot \frac{\ln r}{r} \right\}. \tag*{(\Cref{fact:Tribes-properties} and \Cref{fact: biased tribes influence})}
 \end{align*} 
 On the other hand, for any $j\in [k]$, 
 \begin{align*}  \Inf_{y^{(i)}_j}((f_i)_\pi) &= \Pr\big[\text{$\Tribes_{\ell,\delta}(\bx^{(h,1)}) = 1 \text{ and } \Tribes_{\ell,1-\delta}(\bx^{(h,2)}) = 0$} \big]\cdot \Inf_{j}\big(\Maj_k(y^{(i)})_\pi\big) \\
 &= \delta^2 \cdot \Inf_{j}\big(\Maj_k(y^{(i)})_\pi\big). 
 \end{align*} 
 Since $\Inf_{y^{(i)}_j}((f_i)_\pi) \le \Inf_v((f_i)_\pi)$, the bounds above imply that  
 \[ \Inf_{j}\big(\Maj_k(y^{(i)})_\pi\big) \le \frac1{\delta^2}\cdot \max \left\{(2+o(1))\cdot \delta \log(1/\delta) \cdot \frac{\log \ell}{\ell}, (1+o(1))\cdot \frac{\ln r}{r} \right\}. \] 
 The lemma follows: by choosing $c$ to be a sufficiently small constant in~\Cref{eq: majority most influential}, we can ensure that  $\Inf_{j}\big(\Maj_k(y^{(i)})_\pi\big) \le \frac1{100\sqrt{k}}$. 
 \end{proof}

\begin{lemma}
\label{lem: haha} 
    There is a universal constant $c$ such that the following holds.  Fix $i\in [h]$ and consider a uniform random $\by^{(i)}\in \zo^k$.  The probability there is an input $u$ to $f_h$ consistent with $\by^{(i)}$ such that $T_{\mathrm{approx}}$, on input $u$, queries an $R_i$-variable before a querying at least $ck/\log k$ many $\by^{(i)}$-variables is $O(k^{-2})$.  
\end{lemma}

\begin{proof}
Fix an outcome $y^{(i)}$ of $\by^{(i)}$.  Suppose that there is an input $u$ consistent with $y^{(i)}$ such that $T_{\mathrm{approx}}$, on input $u$, queries an $R_i$-variable after querying only $< ck/\log k$ many $y^{(i)}$-variables.  Call such a $y^{(i)}$ outcome {\sl bad}, and let $\pi$ denote the corresponding path in $T_{\mathrm{approx}}$ that leads to the first query to an $R_i$-variable.  Since $y^{(i)}$-variables mostly precede $R_i$ in $T_{\mathrm{approx}}$, we have that 
\[ \Inf_j\big(\Maj_k(y^{(i)})_\pi\big) \le \frac1{100\sqrt{k}}\quad \text{for all $j\in [k]$.} \] 
For this to hold, it must be the case that among the $t < ck/\log k$ many $y^{(i)}$-variables that occur in $\pi$, the discrepancy between the number of $0$'s and $1$'s is $\Omega(\sqrt{k})$.  We can therefore bound 
\begin{align*} 
\Prx_{\by^{(i)}\in \zo^k}[\text{$\by^{(i)}$ is bad}] &\le \sum_{t=1}^{ck/\log k} \Prx_{\bb \sim \mathrm{Bin}(t,\frac1{2})}\big[|\bb - \lfrac{t}{2}| \ge \Omega(\sqrt{k})\big]  \\
&\le 
\sum_{t=1}^{ck/\log k} e^{-\Theta(k/t)} \tag*{(Hoeffding's inequality)} \\
&\le \frac{ck}{\log k} \cdot e^{-\Theta((\log k)/c)} 
\ll \frac1{k^2},   
\end{align*} 
where the final inequality holds by choosing $c$ to be a sufficiently small constant. 
\end{proof}

By a union bound over $i \in [h]$, we have the following corollary of~\Cref{lem: haha} (which can be thought of as being roughly analogous to~\Cref{lem:long path} in the proof of~\Cref{thm:TD-lower}(b)): 

\begin{corollary}[Most queries to $z$-variables are deep within $T_{\mathrm{approx}}$] 
\label{cor: whee} 
Let $\by = (\by^{(1)},\ldots,\by^{(k)})$ be uniform random.  The probability that there is an input $u$ to $f_h$ consistent with $\by$ such that $T_{\mathrm{approx}}$, on input~$u$, queries a $z$-variable before querying at least $(ck/\log k)\cdot h$ many $\by$-variables is $O(h/k^2)$. 
\end{corollary}

We are finally ready to lower bound the size of $T_{\mathrm{approx}}$: 

\begin{claim}[Lower bound on the size of $T_{\mathrm{approx}}$]
    \label{claim: monotone approx TD is large}
    Fix $\varepsilon \in (0,\frac1{2})$ and let $c = (\frac1{2} - \eps)/2$.  If
    \begin{align} 
    (1-\delta)^{2h} &\ge (2+c)\eps  \label{eq: lots reach tribes}  \\
    h &\le k, \label{eq: bound on h}  
    \end{align} 
    then the size of $T_{\mathrm{approx}}$ is at least $2^{\Omega(hk/\log k)}$. 
\end{claim}

\begin{proof} 
  We will call an input to $f_h$ \emph{$z$-dependent} if it satisfies:
    \begin{align*}
        \Tribes_{\ell,\delta}(x^{(i,1)}) = 0\text{ and }\Tribes_{\ell,1-\delta}(x^{(i,2)}) = 1 \text{ for all $i \in [h]$.}
    \end{align*}
    Note that the output of $f_h$ on any  $z$-dependent input is $\Tribes_r(z)$.  
    Let us define $\zeta(y^{(1)},\ldots,y^{(h)})$ to be the $\zo$-valued indicator of whether there is an input $u$ consistent with $y^{(1)},\ldots,y^{(h)}$ such that $T_{\mathrm{approx}}$ on input $u$ queries a $z$-variable.  For any fixed $y = (y^{(1)},\ldots,y^{(h)})$, 
    \begin{itemize} 
    \item[$\circ$]
    If $\zeta(y) = 0$, then $T_{\mathrm{approx}}$ must assign the same $-1$ or $+1$ value to all $z$-dependent inputs consistent with $y$; 
    \item[$\circ$] The fraction of $z$-dependent inputs that are consistent with $y$ is 
     \[ \Pr\big[ \Tribes_{\ell,\delta}(\bx^{(i,1)}) = 0\text{ and }\Tribes_{\ell,1-\delta}(\bx^{(i,2)}) = 1 \text{ for all $i \in [h]$}\big] = (1-\delta)^{2h}, \] 
    which is at least $(2+c)\eps$ by~\Cref{eq: lots reach tribes}. Furthermore, since output of $f_h$ on any  $z$-dependent input is $\Tribes_r(z)$, among the $z$-dependent inputs that are consistent with $y$, we have that $f_h$ labels half of them $-1$ and half of them $+1$.  
    \end{itemize} 
    Therefore, 
    \[ \error(T_{\mathrm{approx}},f_h) \ge \lfrac1{2}\cdot (2+c)\eps \cdot \Pr[\zeta(\by)=0]. \] 
    Since $\error(T_{\mathrm{approx}},f_h) \le \eps,$ it follows that $\Pr[\zeta(\by)=1] \ge  \Omega(1)$.   Next, applying~\Cref{cor: whee} along with our assumption that $h \le k$ (\Cref{eq: bound on h}), we further have that 
    \[ \Prx_{\by}\Big[ \text{$\exists\,$ $\by$-consistent $u$ s.t.~$T_{\mathrm{approx}}(u)$ queries $z$-variable after $\ge \frac{ck}{\log k}\cdot h$ many $\by$-variables}\Big] \ge \Omega(1). \] 
    On the other hand, for any fixed path $\pi$ in $T_{\mathrm{approx}}$ that queries $\ge \Omega(kh/\log k)$ many $y$-variables, at most a $2^{-\Omega(kh/\log k)}$ fraction of $\by$'s can be consistent with this specific $\pi$.  We conclude that the size of $T_{\mathrm{approx}}$ must be at least 
    $2^{\Omega(kh/\log k)}$, and the proof is complete. 
\end{proof}

\Cref{thm:TD-lower-monotone}(b) now follows from~\Cref{claim: monotone optimal size} and~\Cref{claim: monotone approx TD is large} by setting parameters appropriately: 
\begin{proof}[Proof of~\Cref{thm:TD-lower-monotone}(b)]
    By choosing
    \begin{align*}
        \delta &= \Theta\left(\sqrt[3]{(\log \ell)^{4} \log(1/\eps)/\ell}\right) \\
        k &= \Theta\left(\sqrt[3]{\ell^{4} \log(1/\eps)^{2}/(\log \ell)^4}\right) \\
        r &= \Theta(k) \\
        h &= \Theta\left(\frac1{\delta}\cdot \log(1/\eps)\right),
    \end{align*}
    we satisfy~\Cref{eq: majority most influential,eq: lots reach tribes,eq: bound on h}.  We may therefore apply~\Cref{claim: monotone optimal size} to get that the optimal size of $f_h$ is upper bounded by: 
    \[ \size(f_h) \le \exp\left(O(\sqrt[3]{\ell^{4} \log(1/\eps)^{2}/(\log \ell)^4})\right). \] 
    On the other hand, by~\Cref{claim: monotone approx TD is large}, we have that 
    \[ \TopDownDTsize(f_h,\eps) \ge 2^{\Omega(kh/\log k)} = \exp\left(\Omega(\sqrt[3]{\ell^{5}\log(1/\eps)^4/(\log \ell)^{11}})\right). \] 
    This is a separation of $s$ versus $s^{\tilde{\Omega}(\sqrt[4]{\log(s)})}$. 
\end{proof} 

\begin{remark} 
For our choice of parameters above, we have that $s(n) = \size(f_h) =  2^{\tilde{\Theta}(n^{4/5})}$, where $n  = h(2\ell + k) + r$ is the number of  variables of $f_h$.  A standard padding argument yields the same $s$ versus $s^{\tilde{\Omega}(\sqrt[4]{\log s})}$ separation for any function $s(n) \le 2^{\tilde{O}(n^{4/5})}$. 
\end{remark}

\begin{remark}[Depth separation]
\label{remark: depth separation}
The same proof witnesses a separation of $d$ versus $\tilde{\Omega}_\eps(d^{5/4})$ for the optimal {\sl depth} of $f_h$ versus the depth of the tree that $\BuildTopDownDT(f_h,\eps)$ builds.  This disproves the conjecture of Fiat and Pechyony~\cite{FP04} discussed in~\Cref{sec: our results}, which states that $\BuildTopDownDT$ builds a tree of optimal depth for all monotone functions, even in the case of exact representation ($\eps = 0$). 
\end{remark}

\subsection{Lower bounds for all impurity-based heuristics}

\begin{proposition}[Splitting on the most influential variable of a monotone function maximizes purity gain] 
\label{prop: impurity gain is influence}
Let $f : \{\pm 1\}^n \to \{ \pm 1\}$ be a monotone boolean function.\footnote{For this proof, for notational reasons it will be slightly more convenient  for us to work with $\{\pm 1\}^n$ instead of $\zo^n$ as the domain of $f$.}  Let $\mathscr{G} : [-1,1] \to [0,1]$ be a concave function that is symmetric around $0$, and satisfies $\mathscr{G}(-1) = \mathscr{G}(1) = 0$ and $\mathscr{G}(0) = 1$.   Suppose $i \in [n]$ maximizes: 
 \begin{equation} \mathscr{G}(\E[f]) - \lfrac1{2}(\mathscr{G}(\E[f_{x_i={-1}}])+ \mathscr{G}(\E[f_{x_i=1}])), \label{eq:purity-gain}
 \end{equation} 
 Then $\E[f(\bx)\bx_i] \ge \E[f(\bx)\bx_j]$ for all $j\in [n]$. 
\end{proposition} 

\begin{proof} 
For all functions $f$, not necessarily monotone, $\E[f]$ is precisely the average of $\E[f_{x_i=0}]$ and $\E[f_{x_i=1}]$. Because $\mathscr{G}$ is concave everywhere on its domain, Jensen's inequality ensures that $\mathscr{G}(\E[f])$ is greater than $\lfrac1{2}(\mathscr{G}(\E[f_{x_i=0}])+ \mathscr{G}(\E[f_{x_i=1}])$. Furthermore, again by concavity, we have that this difference increases with the difference between $\E[f_{x_i=1}]$ and $\E[f_{x_i=-1}]$. Therefore the variable $i\in [n]$ that maximizes purity gain (\ref{eq:purity-gain}) also maximizes $|\E[f_{x_i=1}] - \E[f_{x_i={-1}}]|$. 

For a monotone function $f : \{ \pm 1\}^n \to \{\pm 1\}$, we have the following identity for all variables $j\in [n]$: 
\begin{align*} 
\Inf_j(f) &= \Pr[f(\bx) \ne f(\bx^{\oplus j})] \\
&= \Ex[f(\bx)\bx_j] \\
&= \E[f_{x_j=1}] - \E[f_{x_j=-1}].
\end{align*} 
Thus, the variable that maximizes purity gain (\ref{eq:purity-gain}) is also the most influential variable of~$f$.
\end{proof}

Recall that in~\Cref{thm:TD-lower-monotone-any}, we claimed that our lower bound on $\TopDownDTsize(f_h,\eps)$ holds not just for the specific algorithm $\BuildTopDownDT$, but in fact {\sl any} impurity-based top-down heuristic.  To see this, note that in our proof of~\Cref{thm:TD-lower-monotone}(b) described in~\Cref{section: monotone size seperation}, we never used any information about {\sl which} leaf $\BuildTopDownDT$ chooses to split on at each stage, only that when a leaf {\sl is} split, it is replaced by the most influential variable of the corresponding subfunction.  In other words, just like our proof of~\Cref{thm:TD-lower}(b), our proof of~\Cref{thm:TD-lower-monotone}(b) applies not just to the specific tree build by $\BuildTopDownDT(f_h,\eps)$; it in fact lower bounds the size of {\sl any} pruning $T_{\mathrm{approx}}$ of $T_\exact = \BuildTopDownDT(f,\eps = 0)$ that is an $\eps$-approximator to $f_h$. 

By~\Cref{prop: impurity gain is influence}, any tree build by a impurity-based top-down heuristic is a pruning of $T_\exact$, and hence our proof of~\Cref{thm:TD-lower-monotone}(b) extends to establish~\Cref{thm:TD-lower-monotone-any}.

\section{New proper learning algorithms: Proofs of~\Cref{thm:learn-general} and~\Cref{thm:learn-monotone}}
\label{sec:learn}

Recall that \topdown{} builds an approximation to $f$ iteratively. It starts with an empty bare tree $T^\circ$ and repeatedly replaces the leaf with the highest score with a query to that leaf's most influential variable. In section \Cref{subsection: lower bound on score}, we proved lower bounds on the score of the leaf that \topdown{} selects. Using those lower bounds, in section \Cref{subsection: upper bound proofs}, we are able to prove upper bounds on the size of the tree \topdown{} needs to produce to guarantee at most $\varepsilon$ error. If, instead, we only guaranteed that we would pick a leaf with score a fourth of that guaranteed by the lower bounds in \Cref{subsection: lower bound on score}, all of our upper bounds would still hold, up to constant factors in the exponent. In this section, we will show that it is possible to accurately enough estimate influences to guarantee we pick a leaf with score at least a fourth the maximum score. First, we provide a definition of score that takes into account both the leaf and the variable selected.

\begin{definition}[score]
    Given any function $f$, we define the score of a leaf $\ell$ and variable $i$ as follows.
    \begin{align*}
        \mathrm{score}(\ell, i) \coloneqq \Prx_{\bx \sim \zo^n}[\,\text{$\bx$ reaches $\ell$}\,] \cdot  \Inf_{i}(f_\ell)  = 2^{-|\ell|} \cdot \Inf_{i}(f_\ell).
    \end{align*}
\end{definition}

We first show that it is possible to estimate scores sufficiently accurately for monotone functions just from random samples of a function, which proves~\Cref{thm:learn-monotone}.     Let $\bS$ be a random sample from a monotone function $f$, and recall~\Cref{fact:mono-influence}.  We can estimate the score of a particular leaf and variable as follows.
\begin{align*}
    \score(\ell, i,\bS) = \Ex_{\bx, \by \in \bS}\big[\Ind[\text{$\bx$ reaches $\ell$}] \cdot f(\bx)(2\bx_i-1)\big].
\end{align*}
Note that $\Ex_{\bS}[\score(\ell, i,\bS))] = \mathrm{score}(\ell, i)$. Let $t$ be any score threshold and $m$ be the number of examples in $\bS$.  By Chernoff bounds, for any particular leaf $\ell$ and variable $x_i$,
\begin{align*}
    \Prx_{\bS}\big[\score(\ell, i,\bS) \leq \lfrac{t}{2}\big] \leq  e^{-\frac{1}{8}t\cdot m} &\quad \text{if $\score(\ell,i)\ge t$} \\
    \Prx_{\bS}\big[\score(\ell, i, \bS) \geq \lfrac{t}{2}\big] \leq e^{-\frac{1}{12}t\cdot m} &\quad \text{if $\score(\ell,i) \le t/4$}. 
\end{align*}
At step $j$ in \topdown{}, there are $j+1$ leaves in $T^\circ$. If $t$ is the maximum score possible at that step, with probability at least $1 - (j+1)e^{-\frac{1}{12}t\cdot m}$, the leaf and variable with maximum empirical score will have true score at least $\frac{t}{4}$. By \Cref{lemma: minimum score additive}, $\topdown{}(f,\varepsilon)$, at step $j$, there will always be a leaf with score at least $\frac{\varepsilon}{(j+1) \log(s)}$, where $s$ is the decision tree size of $f$. Hence, the maximum empirical score will have true score at least $\frac{1}{4}$ the optimal value with probability at least $1 - (j+1)e^{-\frac{\varepsilon \cdot m}{12(j+1) \log(s)}}$.

The probability that selecting the maximum empirical score is always within $\frac{1}{4}$ of the optimal value for every step from $j=0$ to $j = k-1$ is at least $1 - k^2e^{-\frac{\varepsilon \cdot m}{12k \log(s)}}$. By setting the sample size to
\begin{equation} 
    m = O\left(\frac{k \log s}{\varepsilon} (\log k + \log(1/\delta))\right) \label{eq: definition of m} 
\end{equation} 
with probability at least $1 - \delta$, we choose a sufficiently good leaf for $k$ steps of \topdown{}. Recall that, for monotone functions, \topdown{} builds a decision tree of size at most
\begin{align*}
    k = \min(s^{O(\log(s/\eps)\log(1/\eps))}, s^{O(\sqrt{\log s}/\eps)}).
\end{align*}
Hence, with probability at least $1 - \delta$, taking $\min(s^{O(\log(s/\eps)\log(1/\eps))}, s^{O(\sqrt{\log s}/\eps)})\log(1/\delta)$ is enough to learn to accuracy $\varepsilon$.  This proves~\Cref{thm:learn-monotone}.

If $f$ is not monotone, we cannot accurately estimate influences from just random samples. However, we can estimate influences if given access to {\sl random edges} from $f$.

\begin{definition}[Random edges]
\label{def:random edges} 
    For any function $f: \{0,1\}^n \rightarrow \{\pm 1\}$, a \emph{random edge} is two points of the form $((\bx, f(\bx)), ({\bx}^{\oplus \bi}, f({\bx}^{\oplus \bi})))$, where $\bx \in \{0,1\}^n$ and $\bi \in [n]$ are both picked uniformly at random. A \emph{random edge sample} $\bE$ is a collection of random edges. Given random edge sample $\bE$, we will use $\bE_i$ to refer to all those edges in $\bE$ in which the $i^{\text{th}}$ bit of $x$ is flipped.
\end{definition}
Given a random edge sample $\bE$ of a function $f$, we will be able to accurately estimate influences of the variables in $f$, and learn $f$ using $\topdown{}$. We use the following estimate of score: 
\begin{align*}
    \score(\ell, i,\bE) = \Ex_{((\bx_1,\by_1), (\bx_2,\by_2)) \in \bE_i}\big[\Ind[\text{$\bx_1$ and $\bx_2$ reach $\ell$}] \cdot \Ind[\by_1 \neq \by_2]\big]
\end{align*}
If we desire there to be $m$ samples in each $\bE_i$ with probability at least $1 - \delta$, then having $\bE$ by size $O(n \cdot (m + \log(\frac{1}{\delta}))$ is sufficient, where $m$ is as defined in~\Cref{eq: definition of m}.   Since one can certainly general a random edge sample $\bE$ if given membership query access to $f$, this proves~\Cref{thm:learn-general}.

\pparagraph{Learning trees with small average depth:~\Cref{thm: learning average depth trees}.} Let $f$ be a monotone function computed by a decision tree $T$ of average depth $\triangle$. 
\begin{enumerate}
    \item  We first observe that the total influence of $f$ is at most $\triangle$.  To see this, first recall that $\Inf(f) = \E[\sens_f(\bx)]$, where $\sens_f(x) = |\{ i \in [n] \colon f(x) \ne f(x^{\oplus i})\}|$, i.e.~that total influence is equivalent to average sensitivity.  For any $x$, the sensitivity of $f$ at $x$ is at most the depth of the path that $x$ follows in $T$, and hence the average sensitivity of $f$ is at most the average depth $\triangle$ of~$T$.  
    \item  Recall~\Cref{thm:OS}, which says that monotone functions with decision tree size $s$ have total influence at most $\sqrt{\log s}$. In fact,~\cite{OS07} proves a stronger statement: if $f$ is monotone, then it has total influence at most $\sqrt{\triangle}$.  (This is indeed a stronger statement because $\triangle \leq \log s$.)
    \item Similarly,~\cite{OSSS05} also establishes a stronger version of~\Cref{thm: osss}, showing that $f$ has a variable of influence at least $\Var(f)/\triangle$ (rather than just $\Var(f)/\log s$). Hence an equivalent statement to~\Cref{lemma: minimum score additive} holds, where $\topdown{}$ selects a leaf with score at least $\frac{\varepsilon}{(j+1)\triangle}$.
\end{enumerate}
Combining these observations, with the same proof as for \Cref{thm:TD-upper-monotone}, we get that $\topdown$ produces a tree of size $2^{O(\triangle^2/\varepsilon)}$, and if $T$ is monotone, size only $2^{O(\triangle^{3/2}/\varepsilon)}$. Then, for the same reasons as~\Cref{thm:learn-general,thm:learn-monotone} hold,~\Cref{thm: learning average depth trees} holds.  

\section{Proper learning with polynomial sample and memory complexity}



In this section we give a quasipolynomial-time algorithm for properly learning decision trees under the uniform distribution, where sample and memory of our algorithm are both polynomial.  To our knowledge, this is the first algorithm for properly learning decision trees that achieves polynomial memory complexity.  (Recall~\Cref{table:EH-approx}.) 

\pparagraph{Background: Ehrehfeucht--Haussler and Mehta--Raghavan.} At the core of most learning algorithms is an algorithm that achieves low error on a set of samples. We will use the following notation in this section:\medskip

\noindent \textbf{Notation:} A sample, $S$, is a set of examples of the form $(x,y)$ where $x \in \{0,1\}^n$ and $y \in \{-1,1\}$. The error of a decision tree, $T$, with respect to the samples is
\begin{align*}
    \text{error}_{S}(T) =  \Prx_{\boldsymbol{x}, \boldsymbol{y} \in S}[T(\boldsymbol{x}) \neq \boldsymbol{y}].
\end{align*}
We say that a set of samples is \textit{exactly fit} by a tree of size $s$ if there exists a zero-error tree with size at most $s$. Furthermore, we use $S_{0}^{v}$ and $S_1^{v}$ to refer to all the points in the sample $S$ where the variable corresponding to $v$ is $0$ and $1$ respectively. Lastly, all learning statements in this section are with respect to the uniform distribution.

Ehrenfeucht and Haussler's  algorithm makes the following guarantee:

\begin{theorem}[Algorithmic core of~\cite{EH89}]
\label{thm:EH-core}
    There is an algorithm that takes in a set of samples, $S$, over $n$ variables that can be exactly fit by a decision tree of size $s$ and returns a tree of size at most $n^{\log(s)}$ that exactly fits $S$. Furthermore, that algorithm runs in time $|S| \cdot n^{O(\log s)}$.
\end{theorem}

\noindent One downside of their algorithm is that it leads to a large hypothesis class---the class of all decision trees of size $n^{\log s}$---so in order to generalize with high probability, they require $\text{poly}(n^{\log s}, \frac{1}{\varepsilon})$ samples.

Mehta and Raghavan observe that if a function is computable by a tree of size $s$, then it is also $\varepsilon$-approximated by a tree of depth at most $\log(\frac{s}{\varepsilon})$. They combine this observation with a new algorithm that makes the following guarantee:

\begin{theorem}[Algorithmic core of~\cite{MR02}]
\label{thm:MR-core} 
    There is an algorithm that takes a sample, $S$, over $n$ variables as well as budgets for size $s$ and depth $d$ and returns the decision tree of size at most $s$ and depth at most $d$ with minimal error on $S$.\footnote{They also guarantee that if their are multiple trees with minimal error, they return a tree with minimal size among those with minimal error.} Furthermore, the algorithm runs in time $n^{O(d)} \cdot (s^2 + |S|)$.
\end{theorem}

\noindent  Importantly, there are only $2 \cdot (4n)^s$ decision trees of size at most $s$, a much smaller hypothesis class than for Ehrenfeucht and Haussler's algorithm. As a result, they only need $\text{poly}(s,\frac{1}{\varepsilon}) \cdot \log n$ samples to generalize with high probability. A downside of their work, relative to Ehrenfeucht and Haussler's, is that they need to set $d = O(\log(\frac{s}{\varepsilon}))$, so their algorithm has a runtime of approximately $n^{O(\log(s/\varepsilon))}$ instead of $n^{O(\log s)}$.

Neither \cite{EH89} nor \cite{MR02} are able to learn decision trees with only $\text{poly}(n,s,\frac{1}{\varepsilon})$ memory. \cite{EH89} uses a sample of size approximately $n^{O(\log s)}$ to guarantee generalization, and their algorithm must store all of the samples, so it needs at least that much memory. \cite{MR02} use a dynamic programming algorithm that stores computation for each restriction of the $n$ variables of size at most $d = O(\log(\frac{s}{\varepsilon}))$. There are $\binom{n}{d} \cdot 2^d = \Theta{n \choose \log(s/\eps)}$ such restrictions, resulting in superpolynomial memory complexity. 

\pparagraph{Our algorithm: proper learning with polynomial sample and memory complexity.}  We introduce a new algorithm that makes more assumptions about its input than either \cite{EH89}'s or \cite{MR02}'s algorithms. It requires the samples it receives to be \textit{well-distributed}, a property we will later define (\Cref{def:well-distributed}). In exchange, it only uses polynomial memory. The following should be contrasted with~\Cref{thm:EH-core,thm:MR-core}:

\begin{theorem}[Core of our algorithm] There is an algorithm (\Cref{fig:find}) that when given a depth budget $d$ and a well-distributed sample $S$ that can be exactly fit by a tree of size $s$ returns a tree with depth at most $d$ and error at most $(\frac{3}{4} + o(1))^d \cdot s$ on the samples. Furthermore, the algorithm runs in time $|S| \cdot n^{O(\log s)}$ and uses $\poly(2^d, \log n, |S|)$ memory. \end{theorem}

Note that if the goal is to learn the sample to accuracy $\varepsilon$, we can set $d = O(\log(\frac{s}{\varepsilon}))$. The result is an algorithm that runs in time $|S| \cdot n^{O(\log s)}$ and uses memory $\text{poly}(n, s, 1/\varepsilon, |S|)$. Furthermore, the well-distributed requirement turns out to be true for nearly all uniformly random samples that are sufficiently large. The result is, to the best of our knowledge, the first polynomial memory proper learning algorithm for decision trees.

Our algorithm (\Cref{fig:find}) is a surprisingly simple modification of \cite{EH89}'s algorithm, but our analysis is quite a bit more involved.  A key difference is that~\cite{EH89}'s algorithm is an Occam algorithm, whereas ours is not.  The original \cite{EH89} algorithm breaks down when in cannot fully fit the sample; the analysis showing that our algorithm {\sl is} able to handle a sample it cannot fully fit is subtle.

\begin{figure}[t]
  \captionsetup{width=.9\linewidth}
\begin{tcolorbox}[colback = white,arc=1mm, boxrule=0.25mm]
    \textsc{Find}$(S, s, d)$:
    \begin{itemize}[align=left]
        \item[\textbf{Input:}]  A sample $S$ that can be exactly fit by a tree of size $s$ and depth budget $d$.
        \item[\textbf{Output:}]  A decision tree $T$ with depth at most $d$ that approximately fits $S$. If $S$ cannot be fit by a tree of size $s$, may return ``None."
    \end{itemize}
    \begin{enumerate}
        \item If all samples in $S$ have the same label, return the single-node tree computing that label.
        \item If $s \leq 1$ return ``None." 
        \item \label{Find: depth line} If $d = 0$ return the single-node tree computing the majority label of $S$.
        \item For each relevant\footnote{``relevant" means that neither $S_{0}^v$ nor $S_{1}^v$ is empty} variable $v$:
        \begin{enumerate}
            \item Let $T_{0}^v = \textsc{Find}(S_{0}^v, \frac{s}{2}, d-1)$ and $T_1^v = \textsc{Find}(S_1^v, \frac{s}{2}, d-1)$. 
            \item If both $T_{0}^v$ and $T_1^v$ are not ``None", return the tree with root labeled $v$, $0$-subtree $T_{0}^v$ and $1$-subtree $T_1^v$.
            \item If one of $T_{0}^v$ and $T_1^v$ is ``None" and the other is not:
            \begin{enumerate}
                \item Reexecute the recursive call for the side that was ``None" with size $s-1$ instead of size $\frac{s}{2}$. For example, if $T_1^v$ was ``None", set $T_1^v = \textsc{Find}(S_1^v, s-1, d-1)$
                \item If the reexecuted call still returns ``None", return ``None."
                \item Return the tree with root labeled $v$, $0$-subtree $T_{0}^v$ and $1$-subtree $T_1^v$.
            \end{enumerate}
        \end{enumerate}
        \item Return ``None".
    \end{enumerate}
\end{tcolorbox} 
\caption{Our variant of Ehrenfeucht and Haussler's {\sc Find} algorithm.} 
\label{fig:find}
\end{figure}

\begin{lemma}[Correctness of \textsc{Find}]
    If $S$ can be exactly fit by a tree of size $s$, then \textsc{Find}$(S,s,d)$ will not return ``None." 
\end{lemma}

\begin{proof}

By induction. If $s = 1$ and $S$ can be fit by a tree of size $s$, then all samples in $S$ will have the same label. Hence, \textsc{Find} will return a tree on line 1, and not return ``None".

For $s \geq 2$, there are only two spots where \textsc{Find} could return ``None":
\begin{itemize}[align=left]
    \item[\textbf{Line 4.c.ii}]  \textsc{Find} returns ``None" on line 4.c.ii only if a call of the form $\textsc{Find}(S_a^v, s-1, d-1)$ returns ``None" where $a = \pm 1$ and $v$ is relevant. Let $T_S$ be a minimal size tree that fits $S$, which by assumption, has size at most $s$. Since $v$ is relevant, a node labeled with it must appear somewhere in that tree. This means that there is a size $s-1$ tree that will exactly fit $S_a^v$. By induction, this means that $\textsc{Find}(S_a^v, s-1, d-1)$ will not return ``None."
    \item[\textbf{Line 5.}] \textsc{Find} returns ``None" on line 5 only if there was not a single relevant variable $v$ for which either of the calls 
    $\textsc{Find}(S_{0}^v, \frac{s}{2}, d-1)$ or $ \textsc{Find}(S_1^v, \frac{s}{2}, d-1)$ on line 4.a succeeded (i.e. did not return ``None"). Once again, let $T_{S}$ be a minimal size tree that fits $S$. Then, every node in $T_{S}$ must be relevant for $S$. Furthermore, $T_{S}$ has some root variable $v^*$ and subtrees $(T_{S})_{0}$ and $(T_{S})_{1}$. At least one of $(T_{S})_{0}$ or $(T_{S})_{1}$ must have size at most $\frac{s}{2}$. If $(T_{S})_{0}$ has size at most $\frac{s}{2}$, then by the inductive hypothesis, $\textsc{Find}(S_{0}^v, \frac{s}{2}, d-1)$ does not return ``None." Otherwise $\textsc{Find}(S_{1}^v, \frac{s}{2}, d-1)$ does not return ``None." Hence, \textsc{Find} won't return ``None" on line 5. \qedhere
\end{itemize}
\end{proof}

We hope to prove that \textsc{Find} will produce low error trees, but it turns out to be difficult to guarantee this for arbitrary samples. One particular sample we can guarantee this for is the sample containing all possible points. If $S$ contains all $2^n$ possible points, then \textsc{Find} will return a tree with error at most $\frac{1}{4} \cdot (\frac{3}{4})^d$, which we will prove in \Cref{lemma: error of find on well-distributed samples}. The following property allows us to quantify how close $S$ is to the full sample. 

\begin{definition}[Well-distributed samples]
\label{def:well-distributed} 
    We say that a sample of points $S$ is $c$\textit{-well-distributed} to depth $d$ if, for any restriction $\alpha$ where $|\alpha| \leq d$,
    \begin{equation*}
        ||(S_\alpha)| - \mu| \leq c \mu
    \end{equation*}
    where $\mu = 2^{-|\alpha|} \cdot |\bS|$ is the expected size of $\bS_\alpha$ if $\bS$ is chosen uniformly at random.
\end{definition}

For example, if $S$ contains all possible $2^n$ points, then $S$ is $0$-\textit{well-distributed} to any depth.

\begin{lemma}[Error of \textsc{Find} on well-distributed samples]
\label{lemma: error of find on well-distributed samples}
    Let $S$ be $c$-well-distributed to depth $d$. If $\textsc{Find}(S,s,d)$ does not return ``None," it returns a tree with error at most $\frac{1}{4}(\frac{3}{4} + \frac{c}{4})^d \cdot s$ with respect to $S$.
\end{lemma}
\begin{proof}
    By induction on the $d$; If $d = 0$ and $s \geq 2$, then this lemma requires the error to be less than $\frac{1}{2}$, which $\textsc{Find}$ satisfies since it places the majority node. If $s = 1$ and \textsc{Find} doesn't return ``None," it must have returned a zero-error tree on Line 1, satisfying the desired error bound.
    
    Next, consider $d \geq 1$. If all samples have the same label, \textsc{Find} returns a $0$ error tree. Otherwise, it returns a tree, $T$, with $0$-subtree $T_{0}^v$ and 1-subtree $T_1^v$ for some variable $v$. Let $\ell_{0}$ and $\ell_1$ be the number of points in $S_{0}^v$ and $S_1^v$ respectively. Then, we can relate the errors of the trees as follows:
    \begin{align*}
        \text{error}(T) &= \frac{\ell_{0}}{\ell_{0} + \ell_1}\cdot \text{error}(T_{0}^v) + \frac{\ell_1}{\ell_{0} + \ell_1} \cdot \text{error}(T_1^v)
    \end{align*}
    At least one of $T_{0}^v$ and $T_1^v$ was generated using a recursive call to \textsc{Find} with size parameter $\frac{s}{2}$. Without loss of generality, let that tree be $T_{0}^v$. The other tree, $T_1^v$ was generated by a recursive call with size at most $s$. By the inductive hypothesis,
    \begin{align}
    \label{eq: error in terms of l}
        \text{error}(T) \leq \frac{\ell_{0}}{\ell_{0} + \ell_1}\cdot \frac{1}{4}\left(\frac{3}{4} + \frac{c}{4}\right)^d \cdot \frac{s}{2} + \frac{\ell_1}{\ell_{0} + \ell_1}\cdot \frac{1}{4}\left(\frac{3}{4} + \frac{c}{4}\right)^d \cdot s
    \end{align}

    Since $S$ is $c$-well-distributed to depth $d$, $(1 - c)\mu \leq \ell_{0}, \ell_1 \leq (1+c)\mu$, where $\mu = \frac{|S|}{2}$. Choosing $\ell_{0} = (1 - c)\mu$ and $\ell_1 = (1+c)\mu$ maximizes equation \Cref{eq: error in terms of l} and so results in a valid upper bound.
    \begin{align*}
        \text{error}(T) &\leq \frac{\mu(1 - c)}{2\mu}\cdot \frac{1}{4}\left(\frac{3}{4} + \frac{c}{4}\right)^{d-1} \cdot \frac{s}{2} + \frac{\mu(1 + c)}{2\mu} \cdot \frac{1}{4}\left(\frac{3}{4} + \frac{c}{4}\right)^{d-1} \cdot s \\
        &= \frac{1}{4} \cdot \left(\frac{3}{4} + \frac{c}{4}\right)^{d-1} \cdot \left(\frac{3}{4} + \frac{c}{4}\right) \cdot s \\
        &= \frac{1}{4}\left(\frac{3}{4} + \frac{c}{4}\right)^d \cdot s \qedhere
    \end{align*}
\end{proof}

The above Lemma shows that if a sample is sufficiently well-distributed, \textsc{Find} will return a tree with low error. We next show that, with high probability, sufficiently large samples will be well-distributed.

\begin{lemma}[Well-distributed samples are common]
    \label{lemma: well-distributed common}
    Choose any $0 < c < 1.0, \delta > 0$. Then for
    \begin{equation*}
        m = O\left(\frac{2^d}{c^2} \cdot (d \log(n) + \log(1/\delta))\right),
    \end{equation*}
    a sample of size $m$ chosen uniformly at random from $\{0,1\}^n$ is $c$-well-distributed to depth $d$ with probability at least $1 - \delta$
\end{lemma}
\begin{proof}
    Consider an arbitrary restriction $\alpha$ of length $h \leq d$. By Chernoff bounds,
    \begin{align*}
        \Pr[||\bS_\alpha| - \mu| \geq c\mu] \leq 2e^{-\mu c^2/3}
    \end{align*}
    where $\mu = \mathbb{E}[|\bS_\alpha|] = m \cdot 2^{-h}$. Since $h \leq d$, we can upper bound the probability as follows.
    \begin{align*}
        \Pr[||\bS_\alpha| - \mu| \geq c\mu] \leq 2e^{-m \cdot 2^{-d} c^2/3}
    \end{align*}
    $\bS$ is $c$-well-distributed if $||\bS_\alpha| - \mu| \leq c\mu$ for all possible restrictions $\alpha$ of size at most $d$. There are $\sum_{i=0}^d \binom{n}{i}2^i = n^{O(d)}$ such restrictions. Thus, by a union bound:
    \begin{align*}
        \Pr[\bS\text{ is }c\text{-well-distributed}] \geq 1 - n^{O(d)}\cdot e^{-m \cdot 2^{-d} c^2/3}.
    \end{align*}
    We set the right-hand side of the above equation to be at least $1-\delta$ and solve for $m$, which proves this Lemma.
\end{proof}

Our analysis of the time complexity of \textsc{Find} is very similar to Ehrenfeucht and Haussler's:
\begin{lemma}[Time complexity of \textsc{Find}]
\label{lemma: find time complexity}
    $\textsc{Find}(S,s,d)$ takes time $|S| \cdot (n+1)^{2 \log(s)}$.
\end{lemma}
\begin{proof}
    Fix a total number of variables $n$ and sample size $m$. Let $T(i, s)$ be the maximum time needed by $\textsc{Find}(S,s,d)$ where $S$ has size at most $m$, $i$ is the number of relevant variables in $S$, and $d$ is arbitrary.
    
    If $i = 0$ or $s = 1$, then \textsc{Find} must return on Line 1 or 2, using only $O(m)$ time. Otherwise, the \textsc{Find} makes at most $2i$ recursive calls on line 4.(a) each of which takes time at most $T(i-1, \frac{s}{2})$ It also makes zero or one recursive call on line 4.(c).i which takes time up to $T(i-1, s-1) \leq T(i-1, s)$. In addition to these recursive calls, all of the auxiliary computations can be done in $O(mn)$ time. Hence, we have the following recurrence relation: 
    \begin{align*}
        T(i,s) \leq 2i\cdot T\big(i-1,\lfrac{s}{2}\big) + T(i-1,s) + O(mn).
    \end{align*}
    If we substitute $r = \log(s)$, then equivalently, we have the relation: 
    \begin{align*}
        \tilde{T}(i,r) \leq 2i\cdot \tilde{T}(i-1, r-1) + \tilde{T}(i-1,r) + O(mn).
    \end{align*}
    The above relation is shown to be upper bounded by $\tilde{T}(i,r) = O(m \cdot(n+1)^{2r})$ in \cite{EH89}. Substituting back $r = \log(s)$ gives that $T(i,s) \leq O(m \cdot(n+1)^{2\log(s)})$.
\end{proof}

\begin{lemma}[Memory complexity of \textsc{Find}]
\label{lemma: find memory complexity}
    $\textsc{Find}(S,s,d)$ takes memory $O(2^d(|S| + \log n))$.
\end{lemma}

\begin{proof}
    Each call to \textsc{Find} with depth $d$ will only ever need simultaneously need to run up to $2$ calls to \textsc{Find}, each with depth $d-1$. Hence, there are at most $2^d$ copies of \textsc{Find} that need to be stored in memory at any one time. At worst, each copy stores the sample as well as pointers to it and the $n$ different variables. This means each copy uses $O(|S| + \log n)$ memory, for a total of $O(2^d(|S| + \log n))$ memory.
\end{proof}

The last step in this analysis is a standard generalization argument relying on Chernoff bounds.\begin{lemma}[Generalization]
\label{lemma: generaliztion of find}
    Choose any $\delta,\varepsilon \ge 0$. Suppose that $\boldsymbol{S}$ is a uniformly random sample from a function, $f$, that can be computed by a decision tree of size at most $s$ and depth at most $d$. If the number of samples in $\boldsymbol{S}$ is at least
    \begin{align*}
        m = O\left(\frac{2^d \log(n) + \log(\frac{1}{\delta})}{\varepsilon}\right)
    \end{align*}
    and $\textsc{Find}(\boldsymbol{S},s,d)$ returns a decision tree that fits $\bS$ with error at most $\frac{\varepsilon}{4}$. Then, with probability at least $1 - \delta$, the decision tree returned by $\textsc{Find}$ has error at most $\varepsilon$ on $f$.
\end{lemma}
\begin{proof}
    We will upper bound the number of different decision trees \textsc{Find} could return when given depth limit $d$. There up to $2^d$ spots where a decision tree of depth $ \leq d$ could have a node. In each of these spots, the decision tree could either have one of $n$ variables, a leaf that is either $+1$ or $-1$, or nothing. Thus, the number of decision trees of depth at most $d$ is at most $(n+3)^{2^d}$.
    
    We call a decision tree, $T$, ``bad", if $T$ has error at least $\varepsilon$ on $f$. For any particular bad tree $T$, the probability it will have error less than $\frac{\varepsilon}{4}$ on $\boldsymbol{S}$ can be upper bounded using a Chernoff Bound:
    \begin{align*}
        \Prx_{\boldsymbol{S}}\big[\text{error}_{\boldsymbol{S}}(T) \leq \lfrac{\varepsilon}{4}\big] \leq \exp(-\lfrac{9}{32} \varepsilon m).
    \end{align*}
    Since there are most $(n+3)^{2^d}$ bad trees, the probability that any bad tree has error at most $\frac{\varepsilon}{4}$ is at most $(n+3)^{2^d} \cdot \exp(-\frac{9}{32} \varepsilon m)$. Setting this equal to $\delta$ and solving for $m$ completes the proof of this lemma.
\end{proof}
Finally, we are able to put all these pieces together to prove our main theorem of this section, and show how \textsc{Find} is used to learn decision trees: 
\begin{theorem}[Proper learning with polynomial sample and memory complexity]
\label{thm:our-proper} 
    Let $f$ be any function over $n$ variables computable by a size $s$ decision tree. Choose any $\varepsilon, \delta > 0$. There is an algorithm that
    \begin{itemize}
        \item[$\circ$] runs in time $\poly(n^{\log s},\lfrac{1}{\varepsilon}, \log(\frac{1}{\delta}))$
        \item[$\circ$] requires memory $\poly(s, \log n, \frac{1}{\varepsilon}, \log(\frac{1}{\delta}))$
        \item[$\circ$] uses $\poly(s, \log n, \frac{1}{\varepsilon}, \log(\frac{1}{\delta}))$ random samples from $f$
    \end{itemize} 
    and with probability at least $1 - \delta$ returns a decision tree that is an $\varepsilon$-approximation of $f$.
\end{theorem}
\begin{proof}
    Choose any constant $0 < c < 1$ and set $d = \log(\frac{s}{\varepsilon})/(-\log(\frac{3+c}{4}))$. Then, by taking a uniformly random sample, $\boldsymbol{S}$, of size
    \begin{equation*}
        m = O\left(\frac{2^d}{\varepsilon c^2} \cdot (d \log(n) + \log(1/\delta))\right)
    \end{equation*}
    we have the following holds:
    \begin{enumerate}
        \item $\boldsymbol{S}$ is $c$-well-distributed with probability at least $1 - \frac{\delta}{2}$.
        \item If $\boldsymbol{S}$ is $c$-well-distributed, then $\textsc{Find}(\boldsymbol{S}, s, d)$ returns a tree, $T$, with error at most $\frac{1}{4}(\frac{3 + c}{4})^d \cdot s = \frac{\varepsilon}{4}$ on $\boldsymbol{S}$.
        \item If $T$ has error less than $\frac{\varepsilon}{4}$, then with probability at least $1 - \frac{\delta}{2}$, $T$ is a $\varepsilon$-approximation for $f$.
    \end{enumerate}
    Furthermore, this procedure meets the time constraints by \Cref{lemma: find time complexity} and memory constraints by \Cref{lemma: find memory complexity}.
\end{proof}

\section*{Acknowledgments}

We thank Cl\'ement Canonne, Adam Klivans, Charlotte Peale, Toniann Pitassi, Omer Reingold, and Rocco Servedio for helpful conversations and suggestions.  We also thank the anonymous reviewers of ITCS 2020 for their valuable feedback.

The third author is supported by NSF grant CCF-1921795.

\bibliography{main}{}
\bibliographystyle{alpha}

\end{document}